\long\def\commshort #1\commshortend{}
\long\def\commlong #1\commlongend{#1}

\long\def\oldProbLowerBound #1\oldProbLowerBoundEnd{}

\newcommand{\dist}{\mbox{\rm dist}}

\commshort
\documentclass{llncs}
\commshortend
\commlong
\documentclass[11pt]{article}
\commlongend

\usepackage{cite}
\usepackage[dvips]{graphicx}
\usepackage{amsmath, amssymb, comment}
\usepackage{color}
\usepackage{epsfig,graphicx}
\usepackage[linesnumbered,vlined]{algorithm2e}
\usepackage{enumerate}

\newcommand{\erezdel}[1]{}

\commlong

\newtheorem{theorem}{Theorem}
\newtheorem{thm}{Theorem}[section]

\newtheorem{lemma}[thm]{Lemma}
\newtheorem{claim}[thm]{Claim}
\newtheorem{corollary}[thm]{Corollary}

\newtheorem{definition}[thm]{Definition}
\usepackage{fullpage}

\def\blackslug{\hbox{\hskip 1pt \vrule width 4pt height 8pt
    depth 1.5pt \hskip 1pt}}

\newenvironment{proof}{\begin{trivlist}
\item[\hspace{\labelsep}{\noindent\textit{Proof. }}]
}{\hfill$\Box$\end{trivlist}}

\newcommand{\qedsymb}{\hfill{\rule{2mm}{2mm}}}
\newenvironment{proofof}[1]{\begin{trivlist}
\item[\hspace{\labelsep}{\bf\noindent Proof of #1: }]
}{\qedsymb\end{trivlist}}

\commlongend

\commlong
\def\qed{\quad\blackslug\lower 8.5pt\null\par}
\commlongend

\def\QED{\quad\blackslug\lower 8.5pt\null\par}

\newcommand{\AlgAppDeg}[0]{\mbox{{\sf AppDegree}}}

\newcommand{\SL}[0]{\mathcal{SL}}
\newcommand{\jSL}[0]{j^{\SL}_{max}}
\newcommand{\maxd}[0]{d^*}

\newcommand{\maxSL}[0]{S_{max}^{\SL}}
\def\LeaderElection{\mbox{\sf LeaderElection}}
\newcommand{\SUM}[0]{\mbox{{\sc SUM}}}
\newcommand{\Val}[0]{\mbox{{\sc Val}}}

\def\CompMaxSH{\mbox{\sf CompMaxSH}}
\def\CompMaxMH{\mbox{\sf CompMaxMH}}

\newcommand{\ApproxSize}[0]{\mbox{{\sf ApproxNetSize}}}

\newcommand{\EXP}[0]{\mathbb{E}}
\newcommand{\Prob}[0]{\mathbb{P}}

\newcommand*\samethanks[1][\value{footnote}]{\footnotemark[#1]}

\def\mylevel{\mbox{\tt MYLEVEL}}
\def\maxlevel{\mbox{\tt MAXLEVEL}}
\def\ConstructBFS{\mbox{\sf ConstructBFS}}
\def\id{\mbox{\tt id}}

\begin{document}

\def\ABSTRACT{
\begin{abstract}
This paper studies the theory of the additive wireless network model, in which the received signal is abstracted as an addition of the transmitted signals.
Our central observation is that the crucial challenge for computing in this model is not high contention, as assumed previously, but rather guaranteeing a bounded amount of \emph{information} in each neighborhood per round, a property that we show is achievable using a new random coding technique.
\commlong \par \commlongend
Technically, we provide efficient algorithms for fundamental distributed tasks in additive networks, such as solving various symmetry breaking problems,
approximating network parameters, and solving an \emph{asymmetry revealing} problem such as computing a maximal input.
\commlong \par \commlongend
The key method used is a novel random coding technique that allows a node to successfully decode the received information, as long as it does not contain too many distinct values. We then design our algorithms to produce a limited amount of information in each neighborhood in order to leverage our enriched toolbox for computing in additive networks.
\end{abstract}
}

\commlong

\title{Computing in Additive Networks with Bounded-Information Codes}

\author{
Keren Censor-Hillel
\thanks{Department of Computer Science, Technion, Haifa 32000, Israel. Supported in part by the Israel Science Foundation (grant 1696/14).}
\\ {\sf ckeren@cs.technion.ac.il}
\and
Erez Kantor\thanks{CSAIL, Massachusetts Institute of Technology, MA 01239, USA.
Supported in a part by
NSF Award Numbers CCF-1217506,  CCF-AF-0937274, 0939370-CCF, and AFOSR Contract Numbers  FA9550-14-1-0403 and FA9550-13-1-0042.
}
\\ {\sf erezk@csail.mit.edu}
\and Nancy Lynch\samethanks[2]
\\ {\sf lynch@csail.mit.edu}
\and Merav Parter\samethanks[2]~\thanks{
Merav Parter is also supported by Rothschild and Fulbright Fellowships.
}
\\{\sf parter@csail.mit.edu}
}

\maketitle

\ABSTRACT


\commlongend

\commshort
\title{Computing in Additive Networks with Bounded-Information Codes
\thanks{
The first author is supported in part by the Israel Science Foundation (grant 1696/14).
The last three authors are supported in a part by
NSF Award Numbers CCF-1217506,  CCF-AF-0937274, 0939370-CCF, and AFOSR Contract Numbers  FA9550-14-1-0403 and FA9550-13-1-0042.
Merav Parter is also supported by Rothschild and Fulbright Fellowships.
}
}

\author{Keren Censor-Hillel
\inst{1}
\and Erez Kantor\inst{2}
\and Nancy Lynch \inst{2}
\and Merav Parter \inst{2}
}

\institute{
Department of Computer Science, Technion, Haifa 32000, Israel
\and
CSAIL, Massachusetts Institute of Technology, MA 01239, USA
}

\date{}
\maketitle

\ABSTRACT

\commshortend

\section{Introduction}

The main challenge in wireless communication is the possibility of collisions, occurring when two nearby stations transmit at the same time.
In general, collisions provide no information on the data, and in some cases may not even be distinguishable from the case of no transmission at all. Indeed, the ability to merely detect collisions (a.k.a., the collision detection model) gives additional power to wireless networks, and separation results are known (e.g.,~\cite{SchneiderW2010b}).

Traditional approaches for dealing with interference (e.g., FDMA, TDMA) treat collisions as something that should be avoided or at least minimized~\cite{RadMin,GuKu00,OLT07}. However, modern coding techniques suggest the ability to  \emph{retrieve information} from collisions.
These techniques significantly change the notion of collisions, which now depends on the model or coding technique used.
For example, in \emph{interference cancellation}~\cite{Andrews2005}, the receivers may decode interfering signals that are sufficiently strong and \emph{cancel} them from the received
signal in order to decode their intended message. Hence, from this viewpoint, collision occurs only when neither the desired signal nor the the interfering signal are relatively strong enough.

In this paper, we consider the \emph{additive network model}, in which colliding signals add up at the receiver and are hence \emph{informative} in some cases. It has been shown that such models approximate the capacity of networks with high signal-to-noise ratio~\cite{AvestimehrDT2011}, and that they can be useful in these settings for various coding techniques, such as ZigZag decoding~\cite{GollakotaK08,ParandehGheibiS10}, and bounded-contention coding~\cite{Censor-HillelHLM12}. While in practice there are limitations for implementing such networks to the full extent of the model, the above previous research shows the importance of understanding the fundamental strength of models that allow the possibility of extracting information out of collisions. In a recent theoretical work~\cite{Censor-HillelHLM12}, the problems of local and global broadcast have been addressed in additive networks, under the assumption that the contention in the system is \emph{bounded}.

The central observation of this paper is that in order to leverage the additive behavior of the system, what needs to be bounded is not necessarily the contention, but rather the total amount of \emph{information} a node has to process at a given round. This observation allows us to extend the quantification of the computational power of the additive network model in solving distributed tasks way beyond local and global broadcast. Our key approach in this paper is \emph{not to assume} a bound on the initial number of pieces of information in the system, but rather \emph{guarantee} a bound on the number of \emph{distinct} pieces of information in a neighborhood of every vertex. We then use a new random coding technique, which we refer to as \emph{Bounded-Information Codes (BIC)}, in order to extract the information out of the received signals. This allows us to efficiently solve various cornerstone distributed tasks.

\subsection{Contributions and Methods}
On the technical side, we provide efficient algorithms for fundamental \emph{symmetry breaking} tasks, such as leader election, and computing a BFS tree and a maximal independent set (MIS), as well as algorithms for \emph{revealing asymmetry} in the inputs, such as computing the maximum. We also provide efficient algorithms for approximating network parameters by a constant factor.
Our key methods are based on enriching the toolbox for computing in additive networks with various primitives that leverage the additive behavior of received information and our coding technique.
\paragraph{Main techniques:} The work in~\cite{Censor-HillelHLM12} introduced Bounded-Contention Codes (BCC) as the main technique. BCC allows the decoding of the XOR of any collection of at most $a$ codewords, where $a$ is the bound on the contention. As mentioned, our key approach in this paper is not to assume a bound on the contention, but rather to make sure that the amount of distinct information colliding at a node at a given round is limited.
Our main ingredient is augmenting the deterministic BCC codes with randomization, resulting in Bounded-Information Codes. BIC allows successful decoding of any transmission of $n$ nodes sending at most $O(a)$ distinct values altogether, with high probability.

Randomization plays a key role in the presented scheme in two different aspects. First, the drawback of the standard BCC code is that the transmission of the \emph{same} message by an even number of neighbors is cancelled out. By increasing the message size by factor of $O(\log n)$ and using randomization, BIC codes add random ``noise" to the original BCC codeword so that the probability that two BIC messages cause cancellation becomes negligible.

Another useful aspect of randomization is intimately related to the fact that
our information bounds are logarithmic in $n$. This allows for a win-win situation: if the number of distinct pieces of information (in  a given neighborhood) is small (i.e., $O(\log n)$), the decoding is successful thanks to the BIC codes. On the other hand, if the number of distinct pieces of information is large (i.e., $\Omega(\log n)$), there are sufficiently many transmitting vertices in the neighborhood which allows one to obtain good concentration bounds by, e.g., using Chernoff bounds (for example, in estimating various network parameters).
It is noteworthy that our estimation technique bares some similarity to the well-known \emph{decay strategy}~\cite{BarYeh92} which is widely used in radio-networks. The key distinction between the long line of works that apply this scheme and this paper is the dimension to which this strategy in applied. Whereas so-far, the strategy was applied to the \emph{time} axis (e.g., in round $i$, vertex $u$ transmits with probability $2^{-i}$), here it is applied to the \emph{information} (or message) axis (e.g., vertex $u$ writes the specific information in the $i$'th block of its message with probability $2^{-i}$). This highly improves the time bounds compared to the basic radio model (i.e., the statistics are collected over the multiple blocks of the message instead of over multiple slots).

An immediate application of BIC is a simple logarithmic simulation of algorithms for networks that employ full-duplex radios (where a node can transmit and receive concurrently) by nodes who have only half-duplex radios (where a node either transmits or receives in a given round). This allows us to consider algorithms for the stronger model of full-duplex radios and obtain a translation to half-duplex radios, and also allows us to compare our algorithms to a message-passing setting. To make justice with such comparisons, we note that a message-passing setting not only does not suffer from collisions, but also is in some sense similar to having full duplex, as a node receives and sends information in the same round.

Note that in the standard radio model, collision detection is not an integral part of the model but rather an external capability that can be chosen to be added.
In BIC, collision detection is an integral part of the model, where \emph{collision} now refers to the situation where the number of distinct
pieces of information exceeds the allowed bound. To avoid confusion, the collision detection in the context of BIC, is hereafter referred to as \emph{information-overflow detection}. We show that information-overflow can be detected while inspecting the received codeword, without the need for any additional mechanisms.

\paragraph{Symmetry breaking:}
The first type of algorithms we devise are for various symmetry breaking tasks. The main tool in this context is the \emph{select-level} function, $\SL$, that outputs two random values according to a predefined distribution. Every vertex $v$ computes the $\SL$ function locally, without any communication. The power of this function lies in its ability to assign random \emph{levels} to nodes, such that with high probability\footnote{We use the term \emph{with high probability} (w.h.p.) to denote a probability of at least $1-1/n^c$ for a constant $c\geq 1$.} the maximal level contains at most a logarithmic number of nodes (i.e., below the information bound of the BIC code), and the nodes in the maximal level have different values for their second random variable.

The $\SL$ function allows us to elect a leader in $O(D)$ rounds, w.h.p., where $D$ is the diameter of the network. The elected leader is the node with the maximal pair of values chosen by the $\SL$ function. A by-product of this algorithm is a 2-approximation of the diameter, and the analysis is done over a BFS tree rooted at the leader. We also show how to construct a BFS tree rooted at an arbitrary given node in $O(D)$ rounds, w.h.p, by employing both the $\SL$ function and BIC.

Apart from the above new algorithms, our framework allows relatively simple translations of known algorithms for solving various tasks in message passing systems into additive networks. This includes Luby's MIS algorithm~\cite{Luby86}, Schneider and Wattenhofer's coloring algorithm~\cite{SchneiderW09}, and approximating the minimum dominating set of Wattenhofer and Kuhn~\cite{MDSKhun}, improving significantly over the known bounds for standard radio-model.
We give a flavor of these translations by providing the full MIS algorithm and analysis in
\commlong
Appendix~\ref{subsec:more_sbt},
\commlongend
\commshort
\cite{TR-XOR},
\commshortend
and sketch the results for coloring and approximating the minimum dominating set.

\paragraph{Approximations:}  We design algorithms for approximating various network parameters. We show how to compute a constant approximation of the degree of a node, as well as a constant approximation of the size and diameter of the network. (Our coding scheme only requires nodes to know a polynomial bound $N$ on the network size $n$.) 
Our algorithms naturally extend to solve the more general tasks of local-sum and global-sum approximations\footnote{These are generalizations of degree-approximation and network-size approximation, respectively.} that have been recently considered in \cite{LiuH14}. Yet, the additive setting allows us to obtain much better bounds than those of \cite{LiuH14}.

\paragraph{Asymmetry revealing:} In addition to the above symmetry breaking algorithms, we show that additive networks also allow for fast solutions for tasks which do not require symmetry breaking, but rather already begin with inputs whose asymmetry needs to be revealed: we give an algorithm that computes the \emph{exact} maximal value of all inputs in the network in $O(D\cdot \log{n}/\log\log{n})$ rounds, w.h.p. (in contrast, a $2$-approximation for the maximal value can be computed within $\Theta(D)$ rounds).
We obtain this because our coding scheme allows us to perform a tournament at a high rate. For example, for single-hop networks, in each round only a $O(\log n)$ fraction of the remaining competing vertices survive for the next round.

In some sense, asymmetry revealing can be viewed as the counterpart of symmetry breaking. Clearly, if we compute the maximal input in the system then we can obtain a leader as a by-product. However, the opposite does not hold, and indeed in our leader-election algorithm mentioned above we significantly exploit the fact that the leader need not be predetermined, and use our new toolbox to obtain a leader within only $O(D)$ rounds.

\subsection{Comparison with Related Work}
First, we compare our results with previous theoretical work on the additive network model. The work of~\cite{Censor-HillelHLM12} assumes a bound $a$ on the contention in the system, i.e., there are at most $a$ initial inputs in total in the network. The main method for obtaining global broadcast in the above work is random linear network coding, which can be shown to allow an efficient flow of information in the system. However, this is what requires the bound on the contention. Our BIC coding method bares some technical similarity to the approach of random linear network coding, but allows us to refrain from making assumptions on the total information present in the network.

The aforementioned global broadcast algorithm requires $O(D+a+\log{n})$ rounds. While this algorithm can be used to solve many of the problems that we address in this paper, such as electing a leader and computing the maximal input, it would require $O(n)$ rounds, as for these problems it holds that $a$ can be as large as the total number of nodes in the network. In comparison, our $O(D)$-round leader election algorithm is optimal, and our $O(D\log{n}/\log\log{n})$-round algorithm for computing the maximal input is nearly-optimal, as $O(D)$ is a natural lower bound for both problems, even in the message-passing model.

It is important to mention that our algorithms use messages of size $O(\log^3{n})$. While a standard assumption might be that the message size is $O(\log{n})$ bits, this difference is far from rendering our results easy. In comparison, the global broadcast algorithm of~\cite{Censor-HillelHLM12} requires a message size of $O(a\log{n}+\ell)$ bits for inputs of size $\ell$ and contention bounded by $a$. In our setting, we assume $\ell$ fits the message size (say, is logarithmic in $n$), but since $a$ can be as large as $n$, such a message size would be unacceptable. In addition, if we compare our results to algorithms for the much less restricted message-passing setting, it is crucial to note that even unbounded message sizes do not make distributed tasks trivial. For example, it is possible to compute an MIS in general graphs in $O(\log{n})$ rounds even with messages of size $O(1)$~\cite{MetivierRSDZ2010}, but the best known lower bound is $\Omega(\log{\Delta}+\sqrt{\log{n}})$ even with unbounded messages~\cite{KuhnMW2004}. Recently, Barenboim at el. \cite{BarenboimMIS12} showed a randomized MIS algorithm with $O(\log{\Delta}\cdot\sqrt{\log{n}})$ rounds using unbounded messages.

\commlong
In appendix~\ref{app:addrelatedwork},
\commlongend
\commshort
In \cite{TR-XOR},
\commshortend
we overview results that address the same tasks as this paper in the standard radio network model and in the message-passing model. An additive network can be viewed as lying somewhere in between these two models, as it does suffer from collisions, but to a smaller extent. Nevertheless, while our coding methods assist us in overcoming collisions, the additive network model is still subject to the broadcast nature of the transmissions, and therefore it is highly non-trivial to translate algorithms for the message-passing setting that make use of the ability to send different messages on different links concurrently.
The related work overviewed in
\commshort
\cite{TR-XOR},
\commshortend
\commlong
the aforementioned appendix,
\commlongend
include algorithms and lower bounds for various problems in radio networks, such as the wake-up problem~\cite{Farach-ColtonFM2006}, MIS with and without collision detection~\cite{SchneiderW2010b, MoscibrodaW2005} or with multiple channels
\cite{DaumGGKN13},
leader election~\cite{GhaffariH13}, and approximation of local parameters~\cite{LiuH14}, as well as MIS algorithms for message passing systems~\cite{SchneiderW2010,Luby86,AlonBI1986} and lower bounds~\cite{Linial1992,KuhnMW2004}.

\def\ADDITIONALRELATEDWORK{
In the wake-up problem, nodes can communicate only after successfully receiving a message. Farach-Colton et al.~\cite{Farach-ColtonFM2006} show a lower bound of $\Omega(\log^2{n})$ (more precisely, $\Omega(\log{n}\log{(1/\epsilon)})$ for success probability $\epsilon$) for the number of rounds required for solving the wake-up problem in the standard radio network model. Since sending a single message implies solving the wake-up problem, this gives the same lower bound for MIS. This result holds for a single-hop network with half-duplex radios, no collision-detection, adversarial wake-up, and given only an upper bound on the size of the network. They also give an $\Omega(\log\log{n}\log{(1/\epsilon)})$ lower bound in random geometric graphs where nodes are placed uniformly at random in some area $[0,\ell]^2$. The number of rounds is measured starting from the time at which the first node is woken.

Moscibroda and Wattenhofer~\cite{MoscibrodaW2005} show an MIS algorithm that requires $O(\log^2{n})$ rounds w.h.p. for unit disk graphs in the standard radio network model with half-duplex radios under asynchronous wake-up and no collision-detection. For each node, the number of rounds it requires is measured from the time it is woken until the time it produces an output. The complexity of the algorithm is the maximum taken over all nodes of the number of rounds that they require.

When collision detection \emph{is} available, Schneider and Wattenhofer~\cite{SchneiderW2010b} show that $\Theta(\log{n})$ rounds are required and sufficient for computing an MIS, as well as results about coloring and broadcasting.

Schneider and Wattenhofer~\cite{SchneiderW2010} show an MIS algorithm in $O(\log^{*}{n})$ rounds w.h.p in the classic message-passing model for bounded-independence graphs. These are graphs for which the number of independent nodes in any $r$-neighborhood is bounded by some function $f(r)$. A graph is of polynomially bounded-independence if $f(r)$ is polynomial in $r$. This paper also shows $\Delta+1$-coloring and maximal matching in $O(\log^{*}{n})$ rounds w.h.p. This matches the $\Omega(\log^{*}{n})$ lower bound of Linial~\cite{Linial1992} for MIS in the classic message-passing model, which holds for a ring and therefore also for bounded-independence graphs in general.

For general graphs in the classic message-passing model, the best MIS algorithms are due to Luby~\cite{Luby86} and to Alon et al.~\cite{AlonBI1986}. All algorithms require $O(\log{n})$ rounds w.h.p., while the best known lower bound is of $\Omega(\log{\Delta}+\sqrt{\log{n}})$ due to Kuhn et al.~\cite{KuhnMW2004}.

For standard radio networks another model that was studied is when $F$ channels are available, and a node can choose which channel to transmit on or listen to at any given round.
Daum et al.~\cite{DaumGGKN13} showed an MIS algorithm that requires $O(\frac{\log^2{n}}{F})+\tilde{O}(\log{n})$ rounds w.h.p., and use it to build a constant-degree connected dominating set. They then show how to solve leader election and global broadcast in $O(D+\frac{\log^2{n}}{F})+\tilde{O}(\log{n})$ rounds w.h.p., where $D$ is the diameter of the graph, and $k$-message broadcast in $O(D+k+\frac{\log^2{n}}{F})+\tilde{O}(\log{n})$ rounds w.h.p. The assumptions are that the underlying graph has polynomial bounded-independence, the radios are half-duplex, and no collision-detection is available. The authors also show~\cite{DaumGGKN2013full} a lower bound of $O(\frac{\log^2{n}}{F})+\log{n}$ for the number of rounds required for solving MIS in this model.

The best known algorithms for leader election in the standard radio network model are due to~\cite{GhaffariH13}. When collision detection is available, they provide an algorithm that runs in $O((D+\log{n}\log\log{n})\cdot\min\{\log\log{n},\log{\frac{n}{D}}\})$ rounds, and when collision detection is not available they provide an algorithm that runs in $O((D\log{\frac{n}{D}}+\log^3{n})\cdot\min\{\log\log{n},\log{\frac{n}{D}}\})$ rounds.

Liu and Herlihy~\cite{LiuH14} give algorithms for approximating local sum and global sum in radio networks. Their estimation technique employs the well-known \emph{decay-strategy} \cite{Greenberg87}. Our algorithms for approximating the degree and the size of the network are special case of the local-sum and global-sum respectively. In fact, our algorithms can be slightly modified to solve these tasks.
}

\def\ROADMAP{
\subsection{Roadmap}
The paper is organized as follows. In Section~\ref{sec:background} we give the needed background about additive networks and BCC. In Section~\ref{sec:newtools} we enrich the toolbox for computing in additive networks by introducing BIC codes, simulating full-duplex radios by half-duplex radios, and providing our information-overflow detection scheme. Section~\ref{sec:symbreak} contains our symmetry breaking algorithms, such as leader-election, constructing a BFS tree and computing an MIS, and Section~\ref{sec:approx} consist of our two constant-approximation algorithms, one for the degree of each node and one for the size of the network. Finally, Section~\ref{sec:max} contains our asymmetry revealing algorithm for computing the maximal input in the network, and we conclude with a discussion in Section~\ref{sec:dicussion}.
}

\section{Background: Additive Networks and BCC}
\label{sec:background}

\paragraph{The Additive Network Model:}
A \emph{radio network} consists of stations that can transmit and receive information. We address a synchronous system, in which in each round of communication each station can either transmit or listen to other transmissions. This is called the half-duplex mode of operation. Mainly due to theoretical interest, we also consider the full-duplex mode of operation which is considered harder to implement.
We follow the standard abstraction in which stations are modeled as nodes of a graph $G=(V,E)$, with edges connecting nodes that can receive each other's transmissions.

In the standard radio network model, a node $v\in V$ receives a message $m$ in a given round if and only if in that round exactly one of its neighbors transmits, and its transmitted message is $m$. In the half-duplex mode, it also needs to hold that $v$ is listening in that round, and not transmitting.
If none of $v$'s neighbors transmit then $v$ hears silence, and if at least two of
$v$'s neighbors transmit simultaneously then a {\em collision}
occurs at $v$. In both cases, $v$ does not receive any message.

Some networks allow for {\em collision detection},
where the effect at node $v$ of a collision is
different from that of no message being transmitted, i.e., $v$
can distinguish a collision from silence (despite receiving no message in both).
Other networks operate without a collision detection mechanism, i.e., a node cannot
distinguish a collision from silence. It is known that the ability to detect collisions has a significant impact on the computational power of the network~\cite{SchneiderW2010b}.

In contrast, in this paper, we study the \emph{additive network model}, in which a collision of transmissions is not completely lost, but rather is modeled as receiving the XOR of the bit representation of all transmissions.
More specifically, we model a transmission of a message $m$ by node $v$ as a string of bits. A node $v$ that receives a collision of transmissions of messages $\{m_{u} ~|~ u \in \Gamma(v) \}$, receives their bitwise XOR, i.e., receives the message $y = \bigoplus_{u \in \Gamma(v)}{m_u}$. Here $\Gamma(v)$ is the set of neighbors of $v$. Note that the above notation does not distinguish between the case where a node $u$ transmits to that where it does not, because we model the string of a node that does not transmit as all-zero.

The network topology is unknown, and only a polynomial upper bound $N=n^{O(1)}$ is known for the number of nodes $n$. Throughout, we assume that each vertex $v$ has a unique identifier $\id_v$ in the range $[1, \ldots, n^c]$ for some constant $c\geq 1$. The bandwidth is $O(\text{poly}\log{n})$ bits per message.

\paragraph{Bounded-Contention Coding (BCC):}
Bounded-Contention Codes were introduced in~\cite{Censor-HillelHLM12} for the purpose of obtaining fast local and global broadcast in additive networks. Given parameters $M$ and $a$, a BCC code is a set of $M$ codewords such that the XOR of any subset of size at most $a$ is uniquely decodable. As such, BCC codes can leverage situations where the number of initial messages is bounded by some number $a$, and can be used (along with additional mechanisms) for global broadcast in additive networks. Formally, Bounded-Contention Codes are defined as follows.
\begin{definition}
\label{def:BCC}
An $[M, m, a]$-BCC-code is a set $C \subseteq \{0, 1\}^m$ of size $|C| = M$ such that for any two subsets $S_1,S_2 \subseteq C$ (with $S_1 \neq S_2$) of sizes $|S_1|,|S_2| \leq a$ it holds that $\bigoplus S_1 \neq \bigoplus S_2$.
\end{definition}
Simple BCC codes can be constructed using the dual of linear codes. We refer the reader to~\cite{Censor-HillelHLM12} for additional details and a construction of an $[M, a\log{M}, a]$-BCC code for given values of $M$ and $a$.

\section{New Tools}
\label{sec:newtools}

In this section we enrich the toolbox for computing in additive networks with the following three techniques. The first is a method for encoding information such that it can be successfully decoded not when the number of transmitters in limited, but rather when the amount of distinct pieces of information is limited (even if sent by multiple transmitters concurrently). The second technique is a general simulation of any algorithm for full-duplex radios in a setting of half-duplex radios within a logarithmic number of rounds. Finally, we show that we can detect whether the number of distinct messages exceeds the given threshold.

\paragraph{Bounded-Information Codes (BIC).}
Using BCC and randomization allows one to control the number of distinct pieces of information in the neighborhood.
Let $G=(V,E)$ be an $n$-vertex network and assume that all the messages are integers in the range $[0,n]$. We show that for a bandwidth of size $O(\log^3 n)$, one can use randomization and BCC codes to guarantee that every vertex $v$, whose neighbors transmit $O(\log n)$ \emph{distinct} messages (i.e., hence bounded pieces of information) in a given round, can decode \emph{all} messages correctly with high probability (i.e., regardless of the number of transmitting neighbors).\footnote{The definition of the BIC code can be given for any bound $a$ on the number of distinct values. Since we care for messages of polylogarithmic size, we provide the definition for specific bound $a=O(\log n)$.}
Let $C$ be an $[n, \log^2 n, \log n]$-BCC code and $x \in [0,n]$. By the definition of $C$, the codeword $C(x)=[b_1, \ldots, b_k] \in \{0,1\}^k$ contains $k=O(\log^2 n)$ bits.
Due to the XOR operation, co-transmissions of the same value even number of times are cancelled out. To prevent this, we use a randomized code, named hereafter as a \emph{BIC} code (or BIC for short) as defined next.
\begin{definition}
\label{def:BIC}
Let $C$ be an $[n, \log^2 n, \log n]$-BCC code. An \emph{$[n, c\log^3 n, \log n]$-BIC code for $C$} is a random code $C^I$ defined as follows.
The codeword $C^I(x)$ consists of $k'=\lceil c \cdot \log n \rceil$ blocks, for some constant $c \geq 4$, each block is of size $k=O(\log^2 n)$ (the maximal length of a BCC codeword), and the $i$'th block contains $C(x)$ with probability $1/2$ and the zero word otherwise, for every $i \in \{1, \ldots, k'\}$.
\end{definition}
In other words, for vertex $v$ with value $x$, let $m(v)=C^I(x)$ be the message containing the BIC codeword of $x$ and let $m_i(v)$ denote the $i$'th block of $v$'s message. Then, $m_i(v)=C(x)$ with probability $1/2$ and $m_i(v)=0^k$ 
otherwise.
Let $m'(v)=\bigoplus_{u \in \Gamma(v)}m(u)$ be the received message obtained by adding the BIC codewords of $v$'s neighbors. Then the decoding is performed by using BCC to decode each block
$m'_i(v)$ separately for every $i \in \{1, \ldots, k'\}$, and taking a union over all decoded blocks.
%
%
\def\LEMMABICLOGN{
Let $V' \subseteq V$ be a set of transmitting vertices with values $X'=\bigcup_{v \in V'}\Val(v)$ where $|X'|=O(\log n)$.
For every $v \in V'$, let $C^I_v$ be an $[n, c \cdot \log^3 n, \log n]$-BIC code, for constant $c\geq 4$. Let
$m(v)$ be the $C^I_v$ codeword of $\Val(v)$.
Then, the decoding of $\bigoplus_{v \in V'}m(v)$ is successful with probability at least $1-1/n^{c-1}$. 
}
\begin{lemma}
\label{lem:BIC_logn}
\LEMMABICLOGN
\end{lemma}
\begin{proof}
For every $x \in X'$, let $V_x=\{ v\in V' ~\mid~ \Val(v)=x\}$ be the set of transmitting vertices in $V'$ with the value $x$.
For $x \in X'$ and $i \in \{1, \ldots, k'\}$, let
$V^i_x=\{ v \in V_x ~\mid~ m_i(v)=C(x)\}$ be the set of vertices $v$ whose $i$'th block $m_i(v)$ contains the codeword $C(x)$. We say that block $i$ is \emph{successful} for value $x \in X'$, if $|V^i_x|$ is odd (hence, the messages of $V_x$ are not cancelled out in this block).
Let $M_i \subseteq X'$ be the set of values for which the $i$'th block is successful, and let $V'_i$ contain one representative vertex with a value in $M_i$.
We first claim that with high probability, every value $x \in X'$ has at least one successful block $i_x \in \{1, \ldots, k'\}$. We then show that the decoding of this $i_x$'th block is successful.
The probability that the $i$'th block is successful for $x$ is $1/2$ for every $i \in \{1, \ldots, k'\}$. By the independence between blocks, the probability that $x$ has no successful block is at most $1/n^c$.
By applying the union bound over all $m\leq n$ distinct messages, we get that with probability at least $1-1/n^{c-1}$, every value $x \in X$ has at least one successful block $i_x$ in the message.
Let $m'=\bigoplus_{v \in V'} m(v)$ be the received message and let $m'_i$ be the $i$'th block of the received message. It then holds that
$m'_i=\bigoplus_{v \in V'}m_i(v)=\bigoplus_{v \in V'_i}m_i(v).$
To see this, observe that the values with even parity in the $i$'th block are cancelled out and the XOR of an odd number of messages with the same value $C(x)$ is simply $C(x)$. Since $m'_i$ corresponds to the XOR of $|V'_i|=O(\log n)$ distinct messages, the claim follows by the properties of the BCC code.
\commshort\qed\commshortend
\end{proof}


In our algorithms, the messages may contain several fields (mostly a constant) each containing a value in $[0,n^c]$ for some constant $c\geq 1$.
To guarantee a proper decoding on each field, the messages are required to be aligned correctly. For example, a message containing $\ell$ fields where the $i$'th field contains $x_i \in [0,n]$ is split evenly into $\ell$ blocks and  all bits are initialized to zero. The BIC codeword of $x_i$, denoted by $C^I(x_i)$, is written at the beginning of the $i$'th block. Hence, when the messages are added up, all codewords of a given block are added up separately.
To avoid cumbersome notation, a multiple-field message is denoted by concatenation of the BIC codewords of each field, e.g., the content of a two-field message containing $x_1$ and $x_2$ is referred as $C^I(x_1)\circ C^I(x_2)$, where formally
the message is divided into two equi-length blocks and $C^I(x_1)$ (resp., $C^I(x_2)$) is written at the beginning of the first (resp., second) block.

\paragraph{From full-duplex to half-duplex.}
The algorithms provided in this paper are mostly concerned with the full-duplex setting.
However, in the additive network model, one can easily simulate a full-duplex protocol $P_f$ by half-duplex protocol $P_h$ with a multiplicative overhead of $O(\log{n})$ rounds with high probability%
\commshort
, as explain in more details in \cite{TR-XOR}.
\commshortend
\commlong
, as explain below.
%

Consider a full-duplex protocol $P_f$ in the additive network model. We will describe a half-duplex simulation of $P_f$, denoted by $P_h$. A round $t$ is said to be \emph{successful} for vertex $v$, if $v$ can decode all messages it receives in this round. With BIC codes, a round is successful if $v$'s neighbors sent $O(\log n)$ \emph{distinct} pieces of information.
\def\LEMMAHALFTOFULL{
Each round of a full-duplex protocol $P_f$ can be simulated by half-duplex radios using  $O(\log n)$ rounds, w.h.p. That is, if $t$ is a successful round for $v$ in $P_f$, then $v$ receives all the pieces of information sent to it in this round in $P_f$, within $O(\log n)$ rounds in $P_h$, w.h.p.
}
\begin{lemma}
\label{lemma:half-to-full}
\LEMMAHALFTOFULL
\end{lemma}

\begin{proof}
Consider round $t$ and let $S_t$ be the set of transmitting vertices in $P_f$. Phase $t$ in the half-duplex protocol $P_h$ consists of $O(\log n)$ rounds. In each such round, every vertex $v \in S_t$ chooses to listens or to transmits (if needed) with equal probability. We show that if round $t$ is successful for vertex $v$ in $P_f$, then phase $t$ is successful for vertex $v$ in $P_h$, with high probability.

Let $V_t$ be the set of vertices for which $t$ was a successful round in $P_f$.
Since each vertex $v \in V_t$ receives $O(\log n)$ distinct messages in round $t$ in $P_f$, it implies that there are $O(n \log n)$ communication links $(u,v) \in E$ that need to be satisfied in round $t$.
In each of the $O(\log n)$ rounds in phase $t$,  $u$ transmits and $v$ listens, with probability $1/4$.  Since the set of transmitting stations in each round is a subset of $S_t$ (i.e., $v \in V$ can successfully decode when all the vertices in $S_t$ transmit), $v$ decodes $u$'s message with probability $1/4$ in each round. Since phase $t$ contains $O(\log n)$ rounds, by a Chernoff bound the probability that $v$ did not decode $u$'s message in any of these rounds is at most $1/n^c$. The claim holds by applying the union bound over all $O(n \log n)$ required links.
\commshort\qed\commshortend
\end{proof}
\commlongend

\paragraph{Information-Overflow Detection.} In the standard radio model, a collision corresponds to the scenario where multiple vertices transmit in the same round to a given mutual neighbor. In an additive network, this may not be a problem, since with BIC codes, the decoding is successful as long as there are $O(\log n)$ \emph{distinct} pieces of information in a given neighborhood. In this section, we describe a scheme for detecting an event of information-overflow.
Our scheme is adapted from the contention estimation scheme of~\cite{Censor-HillelHLM12}, designed for the setting of detecting whether there are more than a certain number of initial messages throughout the network. In our setting, the nodes generate values by themselves, and we will later wish to use the fact that we can detect whether too many different values were generated.
The key observation within this context, is that using a BIC code with a doubled information-limit allows one to detect failings with high probability. To see this, assume an information bound $K=c\log{n}$  for constant $c\geq 1$ and consider an $[n, 2K\log n, 2K]$-BCC code $C$.
The BIC code $C^I$ based on $C$ supports $2K$ distinct messages.
\commshort
Throughout, because of space considerations, some of the proofs are omitted.
However, all the proofs are given in the full version \cite{TR-XOR}.
\commshortend

\begin{lemma}
\label{cl:coll_det}
With high probability, either it is detected that the number of distinct values exceeds $K$, or each value $w$ is decoded successfully.
\end{lemma}

\commlong
\begin{proof}
Fix a received codeword, and consider a fixed value $w$ that is sent. For each block $i$, let $X_i$ be the values $z \neq w$ whose parity in the $i$'th block is odd.
If $X_i$ is decodable into more than $K$ values, or if its decoding is illegal then this is detected if the parity of $w$ in that block is even. This happens with probability $1/2$.
Else, $X_i$ is decodable into a set $Q$ of size less than $K$. We claim that if the parity of $w$ is odd, an event which holds with probability $1/2$, then $w$ is successfully decoded regardless of whether $X_i$ is correctly decoded. The reason is that the XOR of $Q$ and $w$ decodes uniquely because it contains at most $K+1$ values and the BCC code supports $2K$ distinct values. This holds even if $Q$ is not the correct set of values included in $X_i$. To summarize, for each value $w$ and for each block $i$, with probability at least $1/2$ either $w$ is decoded or a failure is detected. Since there are $c \cdot \log n$ blocks, the probability the none of these two events happen is at most $1/n^c$. The correctness of the scheme holds by applying the union bound over all $O(n)$ values.
\end{proof}
\commlongend

\section{Symmetry Breaking Tasks}
\label{sec:symbreak}

In this section we show how to solve symmetry breaking tasks efficiently in additive networks. As a key example, we focus on the problem of
leader election.
\commshort
In \cite{TR-XOR},
\commshortend
\commlong
In Appendix \ref{subsec:more_sbt}
\commlongend
we consider additional tasks that involve symmetry breaking such as computing a BFS tree, computing an MIS and finding a proper vertex coloring.
A key ingredient in many of our algorithms is having the vertices choose random variables according to some carefully chosen probabilities, which, at a high level, are used to reduce the amount of information that is sent throughout the network. We refer to this as the $\SL$ (Select Level) function and describe it as follows.

The $\SL$ function does not require communication, and only produces two local random values, an $r$-value and an $z$-value, that can be considered as primary and secondary values for breaking the symmetry between the vertices. The $r$-value is defined by letting $r=j$ with probability of $2^{-j}$, and the $z$-value, $z$, is sampled uniformly at random from the set $\{1,...,2^{8r}\}$.

Note that $\SL$ does not require the knowledge of the number of vertices $n$. We next show that the maximum value of $r(v)$ is concentrated around $O(\log n)$ and that not to many vertices collide on the maximum value. Let $\jSL=\max\{r(v) \mid v\in V \}$ and $\maxSL=\{v\in V \mid r(v)=\jSL\}$.
\def\LEMMASLFUNC{
With high probability, it holds that
(a)  $\jSL \leq 3\log n+1$; (b) $|\maxSL| \leq 2 \log n$; and (c) $z(v)\neq z(v')$ for every $v, v' \in \maxSL$.
}
\begin{lemma}
\label{lem:sl}
\LEMMASLFUNC
\end{lemma}
\begin{proof}
Let $P_v=\Prob(r(v)\geq 3\log n+1)$.
Then, by definition, $P_v=\sum_{i=3\log n+1}^{\infty}2^{-i}
\commshort\\\commshortend =1/n^3$.
By applying the union bound over all vertices in $S$, we get that with probability at least $1-1/n^2$,
$r(v)\leq 3\log n+1$, for every $v\in S$,
as needed for Part (a).
\commlong \par \commlongend
We now turn to bound the cardinality of $\maxSL$.
The random choice of $r(v)$ can be viewed as a random process in which each vertex flips a coin with probability $1/2$ and proceeds as long as it gets ``head".
The value of $r(v)$ corresponds to the first time when it gets a ``tail''.
We now claim that the probability that $|\maxSL|>2 \log n$ is very small.
This holds since the probability that
all of
$2\log n$ coin flips are ``tails''
is exactly $2^{-2\log n}$ which is less than the probability that $|\maxSL|>2 \log n$ and
none of the vertices in $\maxSL$ succeeded in getting another head (and hence in having a larger $r$-value).
Hence,
the probability that $|\maxSL| \leq 2 \log n$
is at least  $1-2^{-2\log n}=1-1/n^2$, as needed for Part (b).

Finally, consider Part (c).
It is sufficient to show that the $z$-values (of vertices of $\maxSL$) are sampled from a sufficient large range.
Note that, the size of this range is $2^{8\cdot\jSL}$.
We later show that $\jSL\geq\log n/2$ with high probability.
This implies that the range size (of the $z$-values) is at least $n^4$ with high probability.
Assume that $\jSL\geq\log n/2$, then the probability  that $z(v)=z(v')$, for any pair $v,v'\in\maxSL$ is at most $1/n^4$.
Applying the union bound over all pairs in $\maxSL$ gives the claim, since $|\maxSL|\leq n$.

In the remaining, we show that indeed,  $\jSL\geq\log n/2$  with high probability.
For every $v\in V$, let $x_v$ be an indicator variable for the event that $r(v)\geq \log n/2$,
i.e., $x_v=1$, if $r(v)\geq \log n/2$ and $x_v=0$, otherwise.
Let $X=\sum_{v\in V} x_v~$.
Note that, the probability that $X\geq 1$ is the same as the probability that $\jSL\geq\log n/2$.
In addition, $\Pr[x_v=1]=2^{-(\log n/2)+1}\geq 2^{-\log n/2}$ and hence (by the linearity of expectation)
$\EXP[X]=\sum_{v\in V} \Pr[x_v=1] = \sqrt{n}$.
By Chernoff bound, the probability that $X= 0$ is exponentially small.
Hence, $X\geq 1$ and so $\jSL\geq\log n/2$ with the high probability.
Part (c) holds.
\commshort\qed\commshortend
\end{proof}

\commshort
\subsection{Leader Election}
\commshortend

\commlong
\subsection{Leader Election}
\commlongend

A Leader-Election protocol is a distributed algorithm run by any vertex such that each node eventually decides whether it is a
leader or not, subject to the constraint that there is exactly one leader. Moreover, at the end of the algorithm all vertices know the $\SL$ function values of the leader.

\commlong
\subsubsection{Leader Election in a single-hop networks}
\commlongend

We first describe a two-round leader election protocol for single-hop networks.
Let $C^I$ be an $[N, O(\log^3 N), O(\log N)]$-BIC code
sampled uniformly at random from the distribution of all random codes that are based on a particular $[N, O(\log^2 N), O(\log N)]$-BCC code $C$ (which is used by all vertices).
First, the vertices apply
the $\SL$ function to compute $r(v),z(v)$.
To do that, in the first communication round, every vertex $v$ transmits $C^I(r(v))$. Since with high probability, by Lemma \ref{lem:sl}(a), $\jSL \leq 2\log n$, the information is bounded and by Claim~\ref{def:BIC}, each vertex can compute $\maxSL$ w.h.p. In the second communication round, every vertex $v$ with $r(v)=\jSL$, transmits $C^I(z(v))$. That is, in the second phase only the vertices of $\maxSL$ transmit the codeword of their $z'$s value.
Since by Lemma \ref{lem:sl}(b), with high probability, $|\maxSL|=O(\log n)$, and by Claim~\ref{def:BIC} again, the $z$-values of all vertices in $\maxSL$ are known to every vertex in the network w.h.p. Finally, the leader is the vertex $v^* \in \maxSL$ with the largest $z$-value, i.e., $z(v^*)=\max_{v' \in \maxSL}z(v')$.
\commshort
In \cite{TR-XOR},
we consider the general case of electing a leader in a network $G$ with diameter $D$, and also show how it implies a 2-approximation of the diameter as a byproduct.
\commshortend

\commlong
\subsubsection{Leader Election in General Networks}
\label{app:leader}

In this section, we consider the general case of network $G$ with diameter $D$. We present Algorithm $\LeaderElection$ that elects a leader within $O(D)$ rounds w.h.p. To enable the termination of the protocol, the vertices compute an approximation for $D$ throughout the course of the leader election process, thereby obtaining a 2-approximation of the diameter is a byproduct of this algorithm.

Initially, every vertex $v$ computes the random values $(r(v),z(v))$ as defined for the single-hop case. In the first two communication rounds, every vertex $v$ transmits the codeword of the maximum $r$-value it has observed so far, and in the third round, if the maximum received $r$-value equals $r(v)$, then it transmits $C^I(r(v)) \circ C^I(z(v))$.

From now on, the algorithm proceeds in \emph{stages}, each stage $t$, consists of four communication rounds, $(t,i)$ for $i \in \{0,1,2,3\}$. The following notation is useful.
For vertex $v$ in stage $t$, let $r_t(v)$ be the maximum $r$-value that $v$ has observed so-far (thus $r_0(v)=r(v)$) and let $z_t(v)$ be the corresponding $z$-value (if received without collisions).
Let $d_t(v)$ be the distance to the vertex $v^*_t$ satisfying that $r_t(v)=r(v^*_t)$ and $z_t(v)=z(v^*_t)$. Hence, the vertex $v^*_t$ can be thought of as the local maximum in the $t$-neighborhood of $v$. Finally, let $\maxd_t(v)$ be the maximum  distance from $v^*_t$, observed by $v$.
To avoid cumbersome notation, we override notation and write $C^I$ whenever a BIC code is in use. Yet, it is important to keep in mind that each application of BIC code, requires an independent sampling of such an instance.

In round $(t,0)$, every vertex transmits the codeword of the maximum $r$-value it has observed so-far, i.e., $C^I(r_t(v))$.
In round $(t,1)$, it transmits $C^I(r_t(v))\circ  C^I(z_t(v))$, if its $r_t(v)$ value is the maximal received.
In round $(t,2)$, it transmits $C^I(r_t(v))\circ C^I(z_t(v)) \circ  C^I(d_t(v))$ if its $d_t(v)$ value is finite. Finally, in round $(t,3)$, every vertex transmits $C^I(r_t(v))\circ  C^I(z_t(v)) \circ C^I(\maxd_t(v))$, if its $\maxd_t(v)$ value is finite.

A vertex that has not receive an update value for none of the fields $\maxd_t(v)$ or $r_t(v)$ for more than $2$ stages, terminates.
This completes the description of the protocol. For a detailed pseudocode, see Figure \ref{fig:lecode}.

\begin{algorithm}
Initially: $(r_0(v),z_0(v))\gets \SL$,
TERMINATE=FLASE\\
Send $C^I(r_0(v))$\\
$r'_0 \leftarrow$ the maximum received value in this round\\
If $r'_0=r_0(v)$ send $C^I(r_0(v)) \circ C^I(z_0(v))$\\
Else, $r_1(v) \leftarrow r'_0$\\
$t \gets t+1$\\
While TERMINATE=FALSE:\\
\quad Round $(t,0)$:\\
\quad\quad Send $C^I(r_t(v))$\\
\quad\quad $r'_t \leftarrow $ the maximum received value\\

\quad Round $(t,1)$:\\
\quad\quad If $r'_t>r_t(v)$ then $r_{t}(v) \leftarrow r'_t$, $d_{t}(v)\leftarrow\infty$, and $\maxd_{t}(v)\leftarrow\infty$\\
\quad\quad Else, send $C^I(r_t(v)) \circ C^I(z_t(v))$\\
\quad\quad If the received $r$-value is $r_t(v)$ then $z_t(v) \leftarrow$ the maximum received $z$-value\\

\quad Round $(t,2)$:\\
\quad\quad If $d_{t}(v)\neq\infty$, send $C^I(r_t(v)) \circ C^I(z_t(v)) \circ C^I(d_t(v))$\\
\quad\quad $r'_t,z'_t,d'_t \leftarrow$ the received values\\
\quad\quad If $r_t(v)=r'_t$ and $z_t(v)=z'_t$ then $d_t(v) \leftarrow d'_t+1$ and $\maxd_{t}(v) \leftarrow d_t(v)$\\

\quad Round $(t,3)$:\\
\quad\quad If $\maxd_{t}(v)\neq\infty$ then send $C^I(r_t(v)) \circ C^I(z_t(v)) \circ C^I(\maxd_t(v))$\\
\quad\quad $r'_t,z'_t,d''_t \leftarrow $ the received values\\
\quad\quad If $r_t(v)=r'_t, z_t(v)=z'_t$ and $\maxd_t(v)<d''_t$ then $\maxd_t(v)\leftarrow d''_t$\\

\quad $t \gets t+1$\\
\quad If $(z_t(v),r_{t}(v))=(z_{t-1}(v),r_{t-1}(v))=(z_{t-2}(v),r_{t-2}(v))$ and $\maxd_t(v)=\maxd_{t-1}(v)=\maxd_{t-2}(v)$\\
\quad\quad\quad\quad then TERMINATE=TRUE\\

    \caption{{\LeaderElection } protocol for vertex $v$.}
    \label{fig:lecode}
\end{algorithm}


\paragraph{Analysis.}
%
Let $v^*$ be the vertex with maximum $r(v^*)$ and $z(v^*)$ values. That is, $v^*$ is the designated leader in the network (global maximum).
Throughout the analysis, we show that with high probability every vertex terminates within $O(D)$ rounds and that  the final leader $\ell(v)$ of every vertex $v$ is the leader $v^*$. It is convenient to analyze the process on the BFS tree rooted at the leader $v^*$. Let $L_i=\{v \in V ~\mid~ \dist(v,v^*,G)=i\}$ be the vertices at distance $i$ from $v^*$, and $\widetilde{L}_i=\bigcup_{j \leq i}L_i$ be the vertices up to distance $i$ from $v^*$ in $G$. For every vertex $v$, let $D_v=\max_{u \in V}\dist(u,v,G)$ denote its local diameter.
We begin by showing the following.
\begin{claim}
\label{cl:le_help}
With high probability, for every stage $t \in \{1, \ldots, D_{v*}\}$, it holds that:\\
(a) after round $(t,0)$, $r_t(v)=r(v^*)$ for every $v \in \widetilde{L}_{t+2}$;\\
(b) after round $(t,1)$, $z_t(v)=z(v^*)$ for every $v \in \widetilde{L}_{t+1}$;\\
(c) after round $(t,2)$, $d_t(v)=\dist(v,v^*)$  and $\maxd_t(v)=\max_{u \in \widetilde{L}_{t}}\dist(u,v^*)$ for every $v \in \widetilde{L}_{t}$.
\end{claim}
\begin{proof}
We prove this by induction on $t$. For the base of the induction, consider $t=1$. In the first communication round, the vertices transmit their $r$-value.  Since there are $O(\log n)$ distinct $r$-values, by Claim~\ref{def:BIC} there are no collisions when transmitting $C^I(r(v))$. Hence, the vertices in $L_1$ (neighbors of $v^*$) know $r(v*)$. In the second communication round, all the vertices of $L_1$ transmit  $C^I(r(v^*))$. Since there are no collisions on $r$-values, the vertices of $L_2$ know $r(v^*)$. In the third round, vertices $v$ whose $r$-value is the maximum $r$-value they observed so far, transmit $C^I(r(v)) \circ C^I(z(v))$.
Hence, the only vertices in $\widetilde{L}_2$ that transmit in the third round, are those that obtain the same $r$-value as $v^*$. By Lemma \ref{lem:sl}(b), there are $O(\log n)$ such vertices and hence there are no collisions at the vertices of $L_1$ and they successfully decode the values of the leader $r(v^*), z(v^*)$. After round $(1,0)$, the vertices of $L_2$ transmit $r(v^*)$
to $L_3$ and since there are no collisions the vertices of $L_3$ know $r(v^*)$. Part (a) of the induction base holds.
In round $(1,1)$, the only vertices in $L_1,L_2,L_3$ that transmit are those that have a $z$-value that corresponds to $r(v^*)$. Hence, by Lemma \ref{lem:sl}(b), there are $O(\log n)$ distinct $z$-values, implying that the vertices of $L_2$ successfully receive $z(v^*)$ from $L_1$. Part (b) of the induction base holds.
In the beginning of round $(1,2)$, the vertices in $L_1,L_2$ know  $(r(v^*),z(v^*))$ but they do not know their distance from $v^*$. Hence, the only transmitting vertex in $\widetilde{L}_2$ is $v^*$, implying that the vertices in $L_1$ successfully receive $((r(v^*),z(v^*),0)$ and hence can set $d_1(v)=1$.
Part (c) of the induction base holds.

Assume the claim holds up to stage $t-1$ and consider stage $t$. By the induction assumption Part (a) for $t-1$, it holds that the vertices in $\widetilde{L}_{t+1}$ know $r(v^*)$.

In round $(t,0)$, the vertices of $L_{t+1}$ transmit $r(v^*)$ and since there are no collisions on this value, the vertices of $L_{t+2}$ know $r(v^*)$ and Part (a) holds.

In round $(t,1)$, the vertices in $L_{t},L_{t+1},L_{t+2}$ transmit $r(v^*)$ and some $z$-value that corresponds to it. Hence, by  Lemma~\ref{lem:sl}(b), as there are $O(\log n)$ distinct $z$-values that correspond to $r(v^*)$, it holds that there are no collisions for the vertices in $L_{t+1}$ and they successfully receive $(r(v^*),z(v^*))$ from the vertices of $L_t$ and Part (b) holds.

In the beginning of round $(t,2)$, the vertices of $L'=L_{t-1} \cup L_{t}\cup L_{t+1}$ know $(r(v^*),z(v^*))$. Hence, the only vertices in $L'$ that transmit are those that have finite distance from $v^*$, namely, $L_{t-1}$. Hence the vertices in $L_{t}$ successfully receive the message $(r(v^*),z(v^*),t-1)$ and by increasing the distance by one, they have the correct distance. Part (c) holds.
%
%
%
\end{proof}

\begin{claim}
\label{cl:le_help2}
With high probability, the following hold:\\
(a) The vertices of $L_{D_{v^*}}$ terminate in stage $D_{v^*}+2$ with the correct values. \\
(b) For every $t \in \{1, \ldots, D_{v^*}\}$, every vertex $v\in L_{D_{v^*}-t}$ knows $D_{v^*}$ after round $(D_{v^*}-t+1,3)$ and terminates in stage $D_{v^*}-t+3$.
\end{claim}
\begin{proof}
Note that in every two stages, as long as the vertex does not obtain the correct values of the leader and the local radius $D_{v^*}$, every vertex either receives an improved maximum distance from its current leader (local maximum in its neighborhood) or a notification of a new leader.
First, observe that by Claim \ref{cl:le_help}(c), after round $(D_{v^*},2)$, $d_{D_{v^*}}(v)=\maxd_{D_{v^*}}(v)=\dist(v,v^*)=D_{v^*}$ for every $v \in L_{D_{v^*}}$. Within two rounds, no update is received on either a new leader (since the global maximum has been found) or on an improved maximum distance, and hence the leaf vertices of $L_{D_{v^*}}$ terminate in stage $D_{v^*}+2$.
Consider Part (b). It is easy to see that for every stage $t\geq 1$, all vertices in $L_{t'}$ for $t' \leq t$ hold the same value of $\maxd_t(v)$. That is for every $t'\leq t$, there exists a value $\ell'$ such that $\maxd_t(v)=\ell'$ for every $v \in L_{t'}$.

We prove the claim by a reversed induction on the stage.
In round $(D_{v^*},3)$, the vertices of $L_{D_{v^*}}$ transmit
$D_{v^*}$ and since the vertices in $L_{D_{v^*}-1}$ receive a message from at most $3$ layers, they can successfully decode the value of $D_{v^*}$. Since they get no update within 2 stages (i.e., they hold the maximum distance from the global maximum vertex), these vertices terminate in stage $D_{v^*}+2$.

Assume that the claim holds up to stage $t'=D_{v^*}-(t+1)$ and consider $t'+1$.  By the induction assumption for stage $t'$, the vertices of $L_{t'}$ receive  $D_{v^*}$ in $(D_{v^*}-t,3)$. Hence in round $(D_{v^*}-t+1,3)$
they transmit $D_{v*}$ to $L_{t'+1}$. Since every vertex in this layer receives a message from at most three distinct layers, they can decode successfully $D_{v^*}$. As there are no future updates, they terminate in stage $D_{v^*}-t+3$, as desired.
\end{proof}

This completes the correctness of the leader-election protocol. Note that this protocol can also be used to compute a $2$-approximation for the diameter of the network.

\commlongend

\def\APPENDDOM{
Using the algorithm of \cite{MDSKhun}, a $O(\log^2 n)$ approximation for the minimum dominating set can be computed within $O(\log n)$ rounds with high probability. Since this simulation is straightforward and does not contain any technical contribution, we defer the details from this extended abstract.
}

\commshort
\section{Approximation Tasks: Degree Approximation}
\commshortend
\commlong
\section{Approximation Tasks}
\commlongend
\label{sec:approx}
In this section we consider approximation tasks.
As a key example, we focus on the task of approximating the degree, i.e., each vertex $v$ is required to compute an approximation for its degree in the graph $G$.
%
%
\commshort
We refer the reader to \cite{TR-XOR} for additional approximation schemes
such as (1) an approximation for the network size; (2) an approximation for the network diameter; and (3) a 2-approximation for the maximum (or minimum).
\commshortend

\label{sec:appdeg}

\commlong
\subsection{Degree Approximation}
\commlongend

We describe Algorithm~$\AlgAppDeg$ that computes with high probability a constant approximation for the degree of the vertices within $O(1)$ rounds.
For vertex $v$ and graph $G$, let $\deg(v,G)=|\Gamma(v,G)|$ be the degree of $v$ in $G$.
When the graph $G$ is clear from the context, we may omit it and simply write $\deg(v)$.
Recall that we assume that each vertex $v$ has a unique identifier $\id_v$ in the range of $[1, \ldots, n^c]$ for some constant $c\geq 1$.

The algorithm consists of two communication rounds (which can be unified into a single round). The first round is devoted for computing the exact degree for low-degree vertices $v$ with degree $\deg(v)\leq c \cdot \log n$. The second round computes a constant approximation for high-degree vertices $v$ with $\deg(v)>c \cdot \log n$. In the first communication round, every vertex $v$ uses a random instance $C^I_v$ of an $[N, c\cdot \log^3 N, c \cdot \log N]$-BIC code to encode its ID and transmits $C^I_v(\id_v)$ as part of $m_1(v)$. In addition, the vertices use the Information-Overflow Detection scheme of Section \ref{sec:newtools} to verify if their BIC decoding is successful (that is, the message $m_1(v)$ consists of two fields, the first encodes the ID and the second is devoted for overflow detection). Upon receiving $m'_{1}(v)=\bigoplus_{u \in \Gamma(v)}m_1(u)$, the vertex applies BIC decoding to the first field of the message and applies Information-Overflow Detection to the second field to verify the correctness of the decoding. Note that by the properties of the BIC code, in this round, the low-degree vertices compute their exact degree in $G$.

The second round aims at computing a constant factor approximation for the remaining vertices with high-degree.
Set $a=40 \cdot \log N$ and $b=2\log N$.
Every vertex $v$ sends an $(a\cdot b)$-bit message $m_2(v)$ defined by a collection of $a$ random numbers in the range of $\{1, \ldots, b\}$ sampled independently by each vertex $v$.
Specifically, for every $v$ and $i \in \{1, \ldots, a\}$, $r_i(v)$ is sampled according to the geometric distribution, letting $r_{i}(v)=j$ for $j \in \{1, \ldots, b-1\}$ with probability $2^{-j}$, and $r_{i}(v)=b$ with probability $2^{-b+1}$ (the remaining probability).
For every $i \in \{1, \ldots, a\}$ and every $j\in \{1, \ldots, b\}$, let $x_{i,j}(v) =1$ if $j< r_i(v)$ and $x_{i,j}(v) =0$ otherwise.
Let $X_i(v)=x_{i,b}(v)\cdots x_{i,2}(v)\cdot x_{i,1}(v)$ and let $m_2(v)=X(v)=X_{a}(v)\cdots X_{2}(v) \cdot X_{1}(v)$ be the transmitted message of $v$.
Let $Y(v)=\bigoplus_{u\in \Gamma(v)}X(u)$ be the received message of $v$.
The decoding is applied to each of the $a$ blocks of $Y(v)$ separately, i.e., treating $Y(v)$ as $Y(v)=Y_{a}(v)\cdots Y_{2}(v) \cdot Y_{1}(v)$, where $Y_i(v)=y_{i,b}(v)\cdots y_{i,2}(v)\cdot y_{i,1}(v)$,
such that
$y_{i,j}(v)=\bigoplus_{u\in \Gamma(v)} x_{i,j}(u)$.
For every $j \in \{1,..., b\}$ and every $v\in V$, define
$\SUM(j,v)=\sum_{i=1}^{a}y_{i,j}(v)$.
Finally, define $j^*(v)=\min\{j\mid \SUM(j,v)\leq 0.2\cdot a\}$,
if there exists an index $j$ such that $\SUM(j,v)\leq 0.2\cdot a$ (we later show that such index do exists with high probability) and $j^*(v)=0$, otherwise as a default value.
The approximation $\delta(v)$ is then given by $2^{j^*(v)-1}$. This completes the description of the algorithm.

As mentioned earlier, the correctness for low-degree vertices follows immediately by the properties of the BIC code and the information-overflow detection (Lemma \ref{lem:BIC_logn} and 
Lemma \ref{cl:coll_det}).
\commshort
We then show that in the second round, for high-degree vertices $v$, we have $\delta(v)/\deg(v)=O(1)$, with high probability.
\commshortend
\commlong
Hence, it remains to show that in the second round, for high-degree vertices $v$, we have $\delta(v)/\deg(v)=O(1)$.
\begin{lemma}
\label{lem:high_degree}
With high probability, if $v$ has high-degree, then  $\delta(v)\in[\deg(v)/5,5\deg(v)]$.
\end{lemma}
\begin{proof}
First, we would like to show that $j^*>1$,
$\SUM(j^*-1,v)\geq 0.2 \cdot a$ and $\SUM(j^*,v) < 0.2 \cdot a$.
Note that the complementary event occurs, if either\\
Case (1) $\SUM(1,v) \leq 0.2 \cdot a$ and then $j^*=1$ and $\SUM(j^*-1=0,v)$ is not defined; or\\
Case (2) $\SUM(j,v) > 0.2 \cdot a$, for every $j\in\{1,...,b\}$ and then $j^*=0$ and $\SUM(j^*=0,v)$ is not defined.
We show that with high probability
%
%
%
%
\begin{eqnarray}
\SUM(b,v)\leq 0.2\cdot a~;
\label{SUM(b,v) leq 0.2 cdot a}
\end{eqnarray}
which, necessarily, implies that case (2) above does not occur, and that
\begin{eqnarray}
\SUM(1,v)>0.2\cdot a~,
\label{SUM(1,v) > 0.2 cdot a}
\end{eqnarray}
which implies that case (10 does not hold as well.
Before proving that inequalities (\ref{SUM(b,v) leq 0.2 cdot a}) and (\ref{SUM(1,v) > 0.2 cdot a}) hold with high probability, we analyze the expectation  of $\SUM(j,v)$.
For every $i \in \{1, \ldots, a\}$ and $j \in \{0,1, \ldots, b\}$, define $V_{i,j}=\{ u ~\mid~  r_i(u)> j\}$ (note that, by this definition, $V_{i,0}=\Gamma(v)$).
For every $j\in \{1,...,b\}$, it then holds that
\begin{eqnarray}
\EXP[\SUM(j,v)] &=& \sum_{i=1}^{a} \Prob\left[ y_{i,j} =1 \right] =\sum_{i=1}^{a} \Prob\left[ |V_{i,j} \cap \Gamma(v)| \mbox{ is odd} \right]  \nonumber \\
&=& \sum_{i=1}^{a} \Prob\big[|V_{i,j-1} \cap \Gamma(v)|\geq 1 \big]/2  \nonumber \\
&=& a \cdot (1-(1-2^{-j+1})^{\deg(v)})/2~.
\label{ineq:lem: E(SUM) = tau (1-) /2}
\end{eqnarray}
Now, we prove that Inequality (\ref{SUM(b,v) leq 0.2 cdot a}) holds with high probability.
Recall that $b=2\log N$ and $\deg(v)\leq n$.
Combining this with Inequality (\ref{ineq:lem: E(SUM) = tau (1-) /2}), we get that
$$
\EXP[\SUM(b,v)]\leq a \cdot(1-(1-2/n^2)^n)/2\leq 0.1a,
$$
where the last inequality holds for $n>10$.
Thus, by Chernoff bound
$$
\Prob[\SUM(b,v) \geq 0.2 \cdot a]\leq\exp(-0.1a/3)\leq 1/n^2~,
$$
as needed for Inequality (\ref{SUM(b,v) leq 0.2 cdot a}), where the last inequality holds, since $a=40\log n$.
We now show that with high probability Inequality (\ref{SUM(1,v) > 0.2 cdot a}) holds too.
Recall that $\deg(v)> 1$.
By Inequality (\ref{ineq:lem: E(SUM) = tau (1-) /2}), we have
$$
\EXP[\SUM(b,1)]=a/2.
$$
Thus, by Chernoff bound
$$
\Prob[\SUM(1,v) < 0.2 \cdot a]\leq\exp(-0.5\cdot a/8)\leq 1/n^2~,
$$
as needed for Inequality (\ref{SUM(1,v) > 0.2 cdot a}), where the last inequality holds, since $a=40\log n$.

As mentions above, Inequalities (\ref{SUM(b,v) leq 0.2 cdot a}) and (\ref{SUM(1,v) > 0.2 cdot a}) implies that $j^*>1$,
$\SUM(j^*-1,v)\geq 0.2 \cdot a$ and $\SUM(j^*,v) < 0.2 \cdot a$.
Intuitively, we use the last two inequalities together with the fact that $\SUM(j^*-1,v)$ and $\SUM(j^*,v)$ concentrating around theirs expectations
to bound from below and bound from above the expectations
$\EXP[\SUM(j^*-1,v)]$ and $\EXP[\SUM(j^*,v)]$, respectively.
%
%
Next, we use the fact that the expectations $\EXP[\SUM(j^*-1,v)]$ and $\EXP[\SUM(j^*,v)]$ are functions of $\deg(v)$ and hence $2^{j^*}$ approximates $\deg(v)$.
Formally, by Chernoff bound, we get, on the one hand, that
\begin{eqnarray*}
\Prob[\SUM(j^*,v) < 0.2 \cdot a \mid \EXP[\SUM(j^*,v)] \geq 0.4\cdot a]
&\leq&\exp(-0.4\cdot a/8)
\leq 1/n^2,
\end{eqnarray*}
where the last inequality holds, since $a\geq 40\log n$.
Therefore,
with high probability,
\begin{equation}
\EXP[\SUM(j^*,v)] \leq 0.4\cdot a~.
\label{ineq:lem: E(SUM) 0.4 tau}
\end{equation}
On the other hand, by Chernoff bound, we have
\begin{eqnarray*}
\Prob[\SUM(j^*-1,v) > 0.1 \cdot a \mid \EXP[\SUM(j^*-1,v)] \leq 0.2\cdot a]
&\leq&\exp(-0.2\cdot a/3)
\leq 1/n^2,
\end{eqnarray*}
where the last inequality holds, since $a\geq 30\log n$.
Thus, with high probability,
\begin{equation}
\EXP[\SUM(j^*-1,v)] \geq 0.1\cdot a ~.
\label{ineq:lem: E(SUM) 0.1 tau}
\end{equation}
Combining Inequality (\ref{ineq:lem: E(SUM) 0.4 tau}) and Inequality  (\ref{ineq:lem: E(SUM) = tau (1-) /2}), we get that
\begin{eqnarray}
(1-2^{-j^*+1})^{\deg(v)}\geq 0.2~,
\label{ineq:lem: (1-1/x)^x geq 0.2}
\end{eqnarray}
and similarly by combining Inequality (\ref{ineq:lem: E(SUM) 0.1 tau}) and Inequality (\ref{ineq:lem: E(SUM) = tau (1-) /2}), we get that
\begin{eqnarray}
(1-2^{-j^*})^{\deg(v)}\leq 0.8~.
\label{ineq:lem: (1-1/x)^x leq 0.8}
\end{eqnarray}
Recall that $(1-1/x)^x<\exp(-1)$, for every $x>0$.
Thus, by combining it with Inequality (\ref{ineq:lem: (1-1/x)^x geq 0.2}),
we have $\exp(-\deg(v)\cdot 2^{-j^*+1})  \geq (1-2^{-j^*+1})^{\deg(v)} \geq  0.2,$
hence
\begin{eqnarray}
\deg(v)\cdot 2^{-j^*} \leq 1.62 ~.
\label{ineq:lem: leq 1.62}
\end{eqnarray}
On the other hand, $(1-1/x)^x>\exp(-1.01)$, for every $x \geq 64$. By combining with Inequality (\ref{ineq:lem: (1-1/x)^x leq 0.8}),
\begin{eqnarray*}
\exp(-1.01 \deg(v)\cdot 2^{-j^*+1})  \leq (1-2^{-j^*+1})^{\deg(v)} \leq  0.8~,
\end{eqnarray*}
yielding,
\begin{eqnarray}
\deg(v)\cdot 2^{-j^*+1} > 0.22~ .
\label{ineq:lem: > 0.22}
\end{eqnarray}
Overall, by Inequalities (\ref{ineq:lem: leq 1.62}) and (\ref{ineq:lem: > 0.22}), we have that
\begin{eqnarray*}
\frac{\deg(v)}{3.24} ~~\leq~~  2^{j^*-1} ~~\leq~~ \frac{\deg(v)}{0.22},
\end{eqnarray*}
as required, and the lemma follows.
\end{proof}
\commlongend
We thus have the following.
\commshort
\begin{theorem}
\commshortend
\commlong
\begin{thm}
\commlongend
\label{thm:approx-deg}
There exists an $O(1)$-round algorithm that computes w.h.p. the exact degree $\deg(v)$ for vertices with $\deg(v)=O(\log n)$ and a constant approximation if $\deg(v)=\Omega(\log n)$.
\commshort
\end{theorem}
\commshortend
\commlong
\end{thm}
\commlongend

\commlong

\subsection{Network Size Approximation}
\label{app:size}
In this section, we present Algorithm \ApproxSize\ that approximates the network size by a constant factor within $O(D)$  rounds w.h.p. At a high level, the nodes choose random variables according to a probability distribution for which we can obtain an approximation for the number of vertices that choose each value. By aggregating the amount of numbers chosen in intervals that grow in size by a constant factor, the nodes try to estimate the size of the network. While intervals with too many values cannot be decoded, they indicate that the number of nodes in the graph is larger. Hence, the nodes look for the interval of largest values in which there are \emph{enough} values so that with high probability indeed it estimates well the total number of nodes.
We obtain the following.
\def\THMNSIZE{
For every network $G=(V,E)$ of diameter $D$, a constant approximation for the size of the network $n=|V|$ can be computed with high probability within $O(D)$ communication rounds.
} 
\begin{thm}
\label{thm:nsize}
\THMNSIZE
\end{thm}

\noindent Recall that it is assumed the vertices know a polynomial bound $N$ on the size of the network (this bound corresponds for example to the message size which is $O(\log n)$). For simplicity, let $N$ be a power of $2$.
The algorithm uses message size of $O(\log^3 N)$ bits and consists of $O(D)$ communication rounds.
To ease the communication scheme, in the first phase, the vertices apply Algorithm $\LeaderElection$ of Section \ref{app:leader} 
to elect a leader $v^*$ and construct a BFS tree $T$ rooted at $v^*$ of depth $D_{v^*}$.

The second phase of the algorithm consists of two parts. The first part is by convergecast of information to the leader $v^*$, which is the root of $T$. In the second part, the leader down casts the result of the computation through $T$ to all the nodes. At a high level, the nodes choose random variables according to a probability distribution for which we can obtain an approximation for the number of vertices that chose each value. By aggregating the amount of numbers chosen in intervals that grow in size by a constant factor, the nodes try to estimate the size of the network. Intervals with too many values chosen cannot be decoded, but will also indicate that the number of nodes in the graph is larger. Hence, the nodes look for the interval of largest values in which there are \emph{enough} values so that with high probability indeed it estimates well the total number of nodes.

Formally, for every $\tau \in \{0,\ldots,D_{v^*}\}$, let $L_{\tau}=\{v ~\mid~ \dist(v,v^*,G)=\tau\}$ be the $\tau'$th level of $T$.
The second phase uses the following codes: an $[N, \log^2 N, \log N]$-BCC code $C$ that is used to encode the IDs of the vertices and an $[N, \log^3 N, \log N]$-BIC code $C^I$ based on $C$ that supports $\log N$ distinct messages.
The message sent by each vertex is divided into $\log N$ blocks, where the $i$'th block is used
to examine the possibility that the number of vertices in the network is $\Theta(2^i)$.
At the beginning of this phase (before communication starts), every vertex $v$, selects a value $\widehat{r}(v)$ according to the geometric distribution such that $\widehat{r}(v)=i$ with probability $p_i=i/2^i$ for $i \in \{1, \ldots, \log N\}$.
In the next $D_{v^*}$ communication rounds, $\tau \in \{1,\ldots, D_{v^*}\}$, the vertices $v$ of $L_{D_{v^*}-\tau+1}$ transmit a message $m(v)$, defined as follows. The message consists of $\log N$ blocks, each of size $\log^2 N$ (i.e., the maximal size of an BIC codeword).
For every $i \in \{1, \ldots, \log N\}$, let
$\widehat{V}_i(v)$ be the IDs of vertices $u$ with $\widehat{r}(u)=i$ that are known to $v$. Initially, $\widehat{V}_{i'}(v)=\{v\}$ for $i'=\widehat{r}(v)$ and $\widehat{V}_{i}(v)=\emptyset$ for every $i \neq i'$.

These sets are updated by $v$ by decoding the message $m'(v)$ received in round $\tau-1$ (in case $v$ is not a leaf).
If $v$ is not a leaf, each of the $\log N$ blocks of the received message $m'(v)$ is decoded  separately by applying the BIC decoding. If the $i$'th block cannot be decoded successfully, it is added to the list $FB(v)$ of failed blocks (blocks that couldn't be decoded). Note that we do not assume that the vertices know if the decoding of a given block is decoded (this can be obtained by using the information-overflow detection scheme of Section~\ref{sec:newtools}, but it is not needed in our algorithm, as will be shown in the analysis).
For every block $i \notin FB(v)$ (i.e., a block that could be decoded), the decoded values (corresponding to vertex IDs) are added to $\widehat{V}_{i}(v)$.
Using the updated $\widehat{V}_{i}(v)$ sets, the $i$'th block of $m(v)$, denoted by $m_{i}(v)$, is defined as follows. If $i \in FB(v)$, then the $i$'block contains the BIC codeword of a special word that indicates failure.
Otherwise, (i.e., the $i$'th block of the received message was decoded successfully), $m_i(v)$ contains the XOR of the BIC codewords of the IDs in $\widehat{V}_{i}(v)$. Formally,
let $C^I_{v,j}$ be random instances of $[N, \log^3 N, \log N]$-BIC codes for $j \in \{1, \ldots, |\widehat{V}_{i}(v)|\}$, then
$m_i(v)=\bigoplus_{id_j \in \widehat{V}_{i}(v)}C^I_{v,j}(id_j)$.
This completes the description of the convergecast communication on $T$.

Let $m'(v^*)$ be the message received by the root vertex $v^*$ and let $i^* \in \{1, \ldots, \log N\}$ be the \emph{last} index $i$ satisfying that $i \notin BF$ and $|\widehat{V}_{i}(v^*)| \geq i/4$.
The estimate $n_{alg}$ for the number of vertices is then $n_{alg}=2^{i^*}$.
Finally, the root vertex $v^*$ downcasts $n_{alg}$ on $T$. This completes the description of the algorithm.
We now turn to the analysis. Clearly, the protocol consists of $O(D)$ communication rounds, so it remains to show correctness.
\begin{lemma}
\label{lem:correct_nsize}
With high probability, $n_{alg}=\Theta(n)$.
\end{lemma}
\begin{proof}
For every $i \in \{1, \ldots, \log N\}$, let $V_i=\{ v ~\mid~ \widehat{r}(v)=i\}$ be the vertices whose $\widehat{r}(v)$ value is $i$.
Let $j^*$ be the index satisfying that $n \in [2^{j^*},2^{j^*+1}]$.

We first claim that with high probability $j \notin \bigcup_{v \in V}FB(v)$ for every $j \in \{j^*, \ldots, \log N\}$. In other words, we show that w.h.p. each of the last $\log N-j^*+1$ blocks are successfully decoded at \emph{each} vertex.
By the properties of the BIC code (see Lemma \ref{lem:BIC_logn}), it is sufficient to show that with high probability $|V_j|\leq 8\log n$ for every $j \in \{j^*, \ldots, \log N\}$, which indeed holds by a Chernoff bound.
As a corollary, we get that $\widehat{V}_{j}(v^*)=\{ID(v) ~\mid~ v \in V_j\}=V_j$ for every $j \in \{j^*, \ldots, \log N\}$. That is, after $v^*$ decodes of $m'(v^*)$, it knows the IDs of the vertices in $V_j$ for every $j \in \{j^*, \ldots, \log N\}$, since the decoding of these blocks was successful all along with high probability.

For $j\geq j^*+6$, the expected number of vertices in $V_j$ is $n \cdot j/2^j\leq j/8$. Since $j \geq \log n$, by a Chernoff bound, we get that with high probability $|V_j|> j/2$ for every $j>j^*+6$.
Since the $j^*$'th block satisfies that $|V_{j^*}|\geq j^*/4$ w.h.p and the root vertex picks the largest index $i^*$ that satisfies this, it holds that $i^* \in \{j^*, \ldots, j^*+6\}$. We therefore have that $n_{alg}=2^{i^*}=\Theta(2^{j^*})=\Theta(n)$. The claim follows.
\end{proof}
Lemma~\ref{lem:correct_nsize}, together with the observation that the required communication is $O(D)$ rounds, gives Theorem~\ref{thm:nsize}.

\subsection{Approximation of Minimum and Maximum
Value in $\Gamma_r(v)$.}
Using BIC codes, one can compute, w.h.p, a constant approximation of any value $x(v) \in \{1, \ldots, n^c\}$ in the $r$-neighborhood within $O(r)$ rounds. This is done by letting $v$ transmit the BIC codeword of the value $j_v$ satisfying that $x(v) \in [2^{j_v-1},2^{j_v}]$. Since there are logarithmic such distinct values, the communication simulates the message passing scheme.
\begin{thm}
\label{lem:appro_poly}
Computing a constant approximation for the minimum or maximum value in the $r$-neighborhood can be done within $O(r)$ rounds with high probability.
\end{thm}
\commlongend

\section{Revealing Asymmetry -- Distributed Tournament}
\label{sec:max}

Consider the setting where every vertex is given an input value (corresponding to its rank, for example) and the goal is to find the vertex with the maximum value.
We will show that BCC codes with message size of $O(\log^3 n)$ allow one to perform many simultaneous competitions between $\Omega(\log n)$ candidates, which result in a tournament process of $O(D\cdot\log n/\log \log n)$ rounds for a network of diameter $D$. Specifically, the fact that the BCC code provides successful decoding when there are $O(\log^2 n)$ concurrent transmitting neighbors, allows us to reduce the number of competitors by a factor of $\Omega(\log n)$ in every round, and hence the winner is found within $O(D\cdot\log n/\log \log n )$ rounds.
\commlong
We begin by describing the protocol for single-hop networks and in Section~\ref{sec:max multi}, we generalize it for any network of diameter $D>1$, which requires some subtle modifications.
\commlongend
\commshort
Because of space considerations, we presenting here only the protocol for single-hop networks.
The protocol for any network of diameter $D>1$, which requires some subtle modifications is presented in the full version \cite{TR-XOR}.
\commshortend

\commlong

\subsection{Exact Computation of the Maximum Value}
\label{app:max}

\subsubsection{Single-hop network}
\commlongend
\commshort
\paragraph{Single-hop network.}
\commshortend
Let $V=\{v_1, \ldots, v_n\}$ be the vertices of the network and let $X=\{x_1, \ldots, x_n\}$, where $x_i \in \{1, \ldots n^2\}$ for all $i$, be the set of integral inputs such that vertex $v_i$ holds the input $x_i$. Let $\max(X)=\max_{i=1}^n x_i$ be the maximum value in $X$.
Note that by Section~\ref{sec:appdeg},
a $2$-approximation for the maximum can be computed within a single round, w.h.p.
The main contribution of this section is the \emph{exact} computation of the maximum value.
\commshort
\begin{theorem}
\commshortend
\commlong
\begin{thm}
\commlongend
\label{thm:exactmax_sh}
The maximum value $\max(X)$ can be computed within $O\left( \frac{\log n}{\log \log n} \right)$ rounds, with high probability.
\commshort
\end{theorem}
\commshortend
\commlong
\end{thm}
\commlongend
Algorithm~\CompMaxSH\ consists of $O(\log n/\log \log n)$ communication rounds. For simplicity, assume that the input values are distinct. This can be obtained by appending to every input value $\lceil \log n \rceil$ least significant bits corresponding to the ID of the vertex.
Let $c\geq 2$ be an upper bound on the approximation ratio of Algorithm~\ApproxSize\ and set $\tau=\lceil c \cdot \log n/\log \log n \rceil$. Initially, all vertices are active. In round $t=\{1, \ldots, \tau\}$, let $n_t$ be a constant approximation for the number of active vertices at the beginning of round $t$, and let $C$ be an $[n_1, 32c \cdot \log^3 n_1, 32c \cdot \log^2 n_1]$-BCC code%
\footnote{This approximation for the size of the network can be obtained by applying Algorithm~\ApproxSize\ or simply Algorithm~\AlgAppDeg\ in the case of single-hop networks (where only the active vertices participate in these algorithms).}.
After computing $n_t$, every active vertex $v_j$ transmits $C(x_j)$ with probability $p_t=4c \cdot \log^2 n_1/n_t$.
If a vertex $v_i$ receives an input $x_j >x_i$ in round $t$, it becomes inactive. The final result $\max(v_i)$ of every vertex $v_i$ corresponds to the maximum input value $x_j$ it received throughout the algorithm.
This completes the description of the algorithm.

We now analyze the algorithm and begin with correctness.
Let $A_t$ be the active vertex set at the beginning of round $t$.
Note that $A_{\tau} \subseteq \ldots \subseteq A_1 =V$.
Let $v_m$ be a vertex with maximum input, i.e., $x_m=\max(X)$.
\begin{lemma}
\label{cl:max_help}
For each  round
$t \in \{1, \ldots, \tau\}$, with high probability it holds that
$|A_t|=O(n_1/\log^{t-1}n_1)$ and $x_m \in A_t$.
\end{lemma}
\begin{proof}
The claim is shown by induction.
For the base of the induction $t=1$, we have that $A_1=V$, and $n_1\leq c \cdot n$ since by the properties of Algorithm~\ApproxSize\ it holds that with high probability $n_1 \in [n/2, c \cdot n]$ for some constant $c\geq 2$. Assume that the claim holds up to step $t-1\geq 1$ and consider step $t$. Order the values of the vertices in $A_{t-1}$ in increasing order of their inputs and consider the subset $H_{t-1} \subset A_{t-1}$ of the $\lceil |A_{t-1}|/\log n_1 \rceil$ vertices with the highest input values in $A_{t-1}$.
We first claim that with high probability, at least one of the vertices in $H_{t-1}$ transmits in round $t-1$.
Since every vertex in $A_{t-1}$ transmits with probability of $p_{t-1}=4c\log^2 n_1/n_{t-1}$ and $n_{t-1}\leq c \cdot |A_{t-1}|$, in expectation there are at least $4\log n_1$ transmitting vertices in $H_{t-1}$ and hence, by a Chernoff bound, w.h.p there is at least one transmitter in $H_{t-1}$.

We proceed by showing that the number of transmitting vertices in round $t-1$ is $O(\log^2 n)$. In expectation, the number of transmitting vertices in $A_{t-1}$ is at most $8c \cdot \log^2 n_1$, and hence by Chernoff bound, with high probability
there are less than $32c\log^2 n_1$ transmitters. By the properties of the BCC code, all messages received in round $t-1$ are decodable.
This implies that all vertices know the value of at least one vertex in $H_{t-1}$ and as a result all vertices in $V \setminus H_{t-1}$ become inactive. In other words, $A_{t} \subseteq H_{t-1}$ and hence $n_t \leq |H_{t-1}|=|A_{t-1}|/\log n_1=O(n_1/\log^{t}n_1)$, where the last equality holds w.h.p by the induction assumption. Finally, by the induction assumption for $t-1$, $v_m \in A_{t-1}$, since all messages were decoded successfully in round $t-1$ w.h.p, it holds that $v_m$ remains in $A_{t}$ as well. The claim follows.
\commshort\qed\commshortend
\end{proof}
We thus have the following, which proves Theorem~\ref{thm:exactmax_sh}.
\begin{lemma}
\label{lem:correct_le_sh}
With high probability $\max(v_i)=\max(X)$ for every vertex $v_i \in V$.
\end{lemma}
\commlong
\begin{proof}
By
Lemma \ref{cl:max_help}, after $\tau$ rounds there are $O(\log n)$ active transmitters w.h.p.
Since the vertex $v_m$ with the maximum input $\max(X)$ remains active in each round, it transmits in the last round, and as its message can be successfully decoded by all vertices, the claim follows.
\commshort\qed\commshortend
\end{proof}
\commlongend

\commlong
\subsubsection{Multi-hop network}
\label{sec:max multi}

The single-hop network protocol can be extended for general networks by paying an extra multiplicative factor of $O(D)$. Notice that this is not immediate from the fact that the diameter is $D$, since it is not clear that concurrent instances of the algorithm for single-hop networks can be run in parallel without incurring an overhead in the number of rounds. Instead, we obtain the result by simulating each round of the algorithm for single-hop networks in a general network in $O(D)$ rounds using a leader and a BFS tree rooted at it.
\begin{theorem}
\label{thm:approx_max_mh}
For any network $G$ of diameter $D$, the maximum value can be computed within $O\left( D \cdot \frac{\log n}{\log \log n} \right)$ rounds, with high probability.
\end{theorem}
\begin{proof}
Algorithm~\CompMaxMH\ operates in a very similar manner to the single-hop case with some modifications. First, the vertices use an $[n_1, \log^4 n_1, \log n_1]$-BIC code (instead of the deterministic BCC code) where $n_1$ is an upper bound on the number of vertices in the network. In addition, the algorithm computes a BFS tree $T$ rooted at some leader $v^*$ using Algorithm~\LeaderElection\ and Algorithm~\ConstructBFS. 

Next, the algorithm proceeds by $\tau=O(\log n /\log \log n)$ phases corresponding to the $\tau$ communication rounds in the single-hop case, each phase consists of $O(D)$ communication rounds. Each phase $t \in \{1, \ldots, \tau\}$ begins by applying Algorithm~\ApproxSize\ to compute the value $n_t$, a constant approximation for the number of active vertices $|A_t|$. Then, every active vertex $v_i \in A_t$ transmits $C^I_{i,t}(x_i)$ with probability $p_t=\log^2 n_1/n_t$ where $C^I_{i,t}$ is a random instance of the $[n_1, \log^4 n_1, \log n_1]$-BIC code. These values are then upcast to the root $v^*$ within $O(D)$ rounds. The root $v^*$ uses BIC decoding and downcasts the maximal value among all the values of the transmitting vertices in $A_t$.
If a vertex $v_i$ receives an input $x_j >x_i$ in phase $t$ then it becomes inactive. The final result $\max(v_i)$ of every vertex $v_i$ corresponds to the maximum input value $x_j$ it received throughout the algorithm. This completes the description of the algorithm, and the correctness follows the same line of argumentation as in the single-hop case.
\end{proof}

\commlongend

\section{Discussion}
\label{sec:dicussion}
It is clear that computing in the additive network model should be doable faster than in the standard radio network model. In this paper we quantify this intuition, by providing efficient algorithms for various cornerstone distributed tasks. Our work leaves open several important open questions for further research.
First, it is natural to ask whether our algorithms can be improved. Specifically, most of our algorithms apply for the full-duplex model and translate into half-duplex by paying an extra factor of $O(\log n)$. It would be interesting to obtain better bounds for half-duplex radios without using the full-duplex protocol as a black box.
An additional axis that requires investigation is the multiple channels model. It would be interesting to study the tradeoff between running time, message size and the number of channels. Note, that whereas most of our algorithms are optimal for  full-duplex radios (up to constant factors), some leave room for improvements. For example, in the problem of computing the maximum input, we believe that some pipelining of the simulation of phases should be able to give a round complexity of $O(D+\log{n}/\log\log{n})$, instead of the current $O(D\cdot\log{n}/\log\log{n})$. However, this is not immediate.
Designing lower bounds for this model is another important future goal. It seems that the problem of computing the maximum input in a single-hop network, should be a good starting point, as we believe that this task requires $\Omega(\log{n}/\log\log{n})$ rounds.
Another interesting future direction involves the
implementation of an abstract MAC layer over \emph{additive}
radio network model. Such an implementation was provided recently \cite{MAC11} for the standard radio network model.
Finally, we note that all our algorithms are randomized, as opposed to the original definition of BCC codes.
Is randomization necessary? What is the computational power of the additive network model without randomization?


\commlong

\clearpage
\pagenumbering{roman}
\appendix
\renewcommand{\theequation}{A-\arabic{equation}}
\setcounter{equation}{0}

\centerline{\large{\bf APPENDIX}}
\numberwithin{equation}{section}

\section{Additional Related Work}
\label{app:addrelatedwork}
\ADDITIONALRELATEDWORK

%

\section{Additional Symmetry Breaking Tasks}
\label{subsec:more_sbt}

\subsection{Construction of BFS Trees}
\label{app:bfs}

\def\APPENDBFSTHM{
For every network $G=(V,E)$ of diameter $D$ and a source vertex $s \in V$, a BFS tree rooted at $s$ can be constructed with high probability within $O(D)$ communication rounds.
} 

In this section, we consider the construction of a Breadth-First-Search (BFS) tree rooted at a given source vertex $s$. Towards the end of this section, we show the following.
\begin{thm}
\label{thm:bfs}
\APPENDBFSTHM
\end{thm}

Note that one can combine Theorem \ref{thm:bfs} and the leader-election protocol to construct a spanning tree of radius $O(D)$ within $O(D)$ rounds w.h.p. (i.e., by computing first a leader in $O(D)$ rounds and then constructing a BFS tree with respect to this leader).

Let $D_{s}=\max_{v \in V}\dist(s,v,G)$ be the radius of the BFS tree and define $L_t=\{v \in V ~\mid~ \dist(s,v,G)=t\}$ as the vertices at distance $t$ from $s$, for every $t \in \{0,\ldots, D_s\}$. Recall that we assume that each vertex $v$ has a unique identifier $\id_v$ in the range of $[1, \ldots, n^c]$ for some constant $c\geq 1$.

Algorithm \ConstructBFS\ consists of two phases.
The first phase consists of $O(D)$ stages, during which the vertices compute their \emph{level} in the BFS tree rooted at $s$ (i.e., distance from $s$). To detect termination, the vertices also compute the diameter $D_s$. The second phase consists of $3$ communication rounds, and is devoted for selecting a unique parent for each vertex.
Each of the $O(D)$ stages of the first phase consists of two alternating rounds in which two types of messages, \maxlevel\ and \mylevel\ are sent by the vertices. In the odd round of stage $t\geq 0$, the vertices of $L_t$, transmit a \mylevel\ message consisting of the $C^I_{t,v,1}$ codeword of their level where each $C^I_{t,v,1}$ is a randomly sampled instance of $[N, \log^3 N, \log N]$-BIC code.
Initially, every vertex $v$, sets $d_v=\infty$. Upon receiving the first \mylevel\ message $d'$ it lets $d_v \gets d'+1$ and it becomes active in the odd round of the next stage.
In the even round of stage $t\geq 0$, \emph{every} vertex $v$ uses $C^I_{t,v,0}$, a randomly sampled instance of $[N, \log^3 N, \log N]$-BIC code. It transmits a \maxlevel\ message consisting of the $C^I_{t,v,0}$ codeword of the value $d^*_t(v)$ where $d^*_t(v)$ is the \emph{maximum} value of all previous \maxlevel\ messages, where $d^*_0(v)$ is initialized to $0$. If the root vertex $s$ did not receive a \maxlevel\ message with an increased value for more than two stages, it initiates a termination message.
Upon receiving a termination message, a vertex $v \in L_t$, waits for $D_s-t$ rounds before beginning the second phase.

The second phase aims to break the symmetry between potential parents of a given vertex. To provide a separation between conflicting levels in the BFS tree, the phase consists of three communication rounds. The vertices of the level $L_t$ transmit an  $O(\log^3 n)$-bit message $m_v$ in round $t \mod 3$. Let $C$ be an $[n, \log^2 n, \log n]$-BCC code that is used to encode the IDs of the vertices. The $O(\log^3 n)$ bits message $m_v$ is divided into $2\log n$ blocks each of size $O(\log n)$. The vertex $v$ writes the codeword of its ID in the $j$'th block with probability of $2^{-j}$. Formally, let $r(v)$ be the $r$-value computed by Protocol $\SL$.  Then, $v$ writes $C(\id_v)$ in the $r(v)$'th block.

Upon receiving a message in round $(t-1)\mod 3$, a vertex $u \in L_t$ decodes the last occupied block of the received message (in the analysis we show that the decoding is successful); it then selects one of the decoded IDs as its parent.
This completes the description of the algorithm.

\paragraph{Analysis.}
We begin by analyzing the first phase of the algorithm and show that the vertices successfully compute their level in the BFS tree and the diameter.
\begin{claim}
\label{cl:bfs_firstround}
With high probability, the following hold:
\begin{enumerate}
\item After the odd round of stage $t$, $d_v=t$ for every $v \in L_t$, $t \in \{1, \ldots, D_s\}$.

\item  After the even round of stage $2D_s-t$, $d^*_t(v)=D_s$ for every vertex $v \in L_t$.

\item  The root $s$ initiates a termination message in stage $2D_s+2$.
\end{enumerate}
\end{claim}
\begin{proof}
Part (1) is shown by induction on $t$. For the base of the induction, consider $t=1$. In the odd round of stage $t=1$, the root vertex $s$ transmits $0$ and thus the vertices of $L_1$ successfully decode this value and by letting $d_v=1$ they hold the correct distance. Assume the claim holds up to stage $t-1$ and consider stage $t$. The only active vertices in stage $t$ are $L_{t-1}$ and by the induction assumption $d_v=t-1$ for every $v \in L_{t-1}$ at the beginning of stage $t$. Hence, by the properties of the BIC code, the vertices of $L_{t}$ successfully decode the message \mylevel\ that contains the value $t-1$ and $d_v=t$ for every $v \in L_{t}$ as desired.

Consider Part (2). It is easy to see to that $d^*_t(v)=d^*_t(v')$ for every $v, v' \in L_t$ for every level $L_t$. In addition, observe that within every two stages, the root $s$ gets a \maxlevel\ message with an increased value. By Part (1), in stage $D_s$, the vertices of level $D_s$ know their level. Since the vertices of each level transmit the same \maxlevel\
message, every vertex receives at most three distinct \maxlevel\ values in a single round, and by the properties of the BIC code, the message is decoded successfully.  The claim can be shown by a backwards induction on the stage $t$, in a similar manner to Part (1).
Finally, by Part (2), $s$ receives $D_s$ in stage $2D_s$ and hence terminates in stage $2D_s+2$, during these two rounds no \maxlevel\ message with an improved value is received and hence it initiates termination.
\end{proof}
We now consider the second phase. Note that by the waiting time defined for every vertex upon receiving the termination message from $s$, the vertices are synchronized at the second phase.
To establish correctness, it is left to show that with high probability, the message size is sufficient and the for every vertex $v$,  the last occupied block of each received message is decodable. Recall that we define $\jSL=\max\{r(v) \mid v\in V \}$ and $\maxSL=\{v\in V \mid r(v)=\jSL\}$.
\begin{claim}
\label{lem:bfs_cor}
Consider the vertices of $L_t$ for some $t \in [0,\ldots, D_s]$. W.h.p., the following hold:
\begin{enumerate}[(a)]
\item  The only transmitting neighbors of $v \in L_t$ in round $(t-1)\mod 3$ are in $L_{t-1}$.

\item $\jSL\leq 2\log n$;

\item The decoding of the last occupied block is successful;
\end{enumerate}
\end{claim}
\begin{proof}
Part (a) follows by definition where in round $i$ such that $(t-1)\mod 3=i$ only the parenting level $L_{t-1}$ transmits and the levels $L_t$ and $L_{t+1}$ are silence. This implies that the messages received by the vertices of $L_t$ in this round are sent by the parenting level. Consider (b).
By Lemma \ref{lem:sl}, 
with high probability,
$\jSL \leq 2\log n$ and $|\maxSL| \leq 2 \log n$.
Finally, consider (c). By the proof of Lemma \ref{lem:sl}(b), the number of parents that wrote into the last occupied block is at most $2\log n$, and by the properties of the BCC code (and the uniqueness of the ID's), the decoding is successful. 
\end{proof}

\subsection{MIS Computation}
\label{app:mis}

In this section we discuss algorithms for finding an MIS in the network. That is, each node has to output a value in $\{0,1\}$ such that the set of nodes that output $1$ is a maximal independent set in the graph.
We address both general graphs and graphs with bounded-growth. A graph with bounded growth is a graph for which there is a function $f(r)$ such that the number of independent nodes in every $r$-neighborhood is at most $f(r)$. Graphs of bounded growth have been studied in the literature for the standard radio network model since intuitively one expects a real wireless network to be such that stations that are close to some transmitter are also relatively close to each other. Algorithms for computing an MIS in a message passing model were given in~\cite{KuhnMW05,MISR05,SchneiderW2010}, with an optimal algorithm requiring $O(\log^*{n})$ rounds~\cite{SchneiderW2010}. In the standard radio network model with collision detection, an MIS algorithm using $O(\log{n})$ rounds was given~\cite{SchneiderW2010b}.

We claim that for graphs of bounded growth, one can compute an MIS in the additive network model within $O(\text{poly}\log\log{n})$ rounds with full-duplex radios.\footnote{Using our simulation, this implies $\tilde{O}(\log{n})$ rounds with half-duplex radios, but the state-of-the-art algorithm for the standard radio network model with half-duplex radios requires only $O(\log{n})$ rounds~\cite{SchneiderW2010b}.}
The main tool that is required by the algorithm is the degree approximation procedure of Section ~\ref{sec:approx}, and then one can essentially simulate the algorithm of~\cite{GfellerV07}, with a similar analysis. We omit the details from this extended abstract.

\subsubsection*{Computing MIS for general graphs in full-duplex model}
We show that for general graphs, it is possible to find an MIS in an additive wireless network within $O(\log{n})$ rounds, w.h.p. This matches the best known algorithm for a message-passing setting, due to Luby~\cite{Luby86} and Alon et al.~\cite{AlonBI1986}. In fact, we show that in the additive network model we can simulate Luby's algorithm efficiently. We begin by recalling Luby's algorithm, and afterwards we explain the challenges for implementing it in a wireless network. Finally, we describe how we overcome these challenges and present our implementation and its proof.

Luby's cornerstone algorithm works in phases, where in each phase every node $v$ chooses to mark itself with probability $1/2d(v)$. If a marked node has the largest degree within its marked neighborhood (ties broken arbitrarily, say, by ID), then it enters the MIS in this phase and is removed from the graph along with all of its neighbors. The following phase is executed with the remaining graph (and remaining degrees). It is straightforward that this algorithm produces an MIS, since no two neighbors can enter the MIS at the same phase, and all neighbors of an MIS node are removed along with it in the phase in which it entered the MIS. The beauty of the algorithm lies in its ability to remove a constant fraction of the edges from the graph in every phases, implying that it completes after $O(\log{n})$ rounds, w.h.p.

To implement Luby's algorithm in an additive network, we need the following tools. First, since a node marks itself with probability inversely proportional to its degree, it has to able to compute its degree in the remaining subgraph. For this, we use our degree-approximation technique, and prove that working with approximate degrees is sufficient.  Second, a node has to know whether it has the maximal estimated degree within its marked neighborhood. While we can compute an approximation to the maximal value in a neighborhood efficiently, here an approximation is insufficient: it may be that there is more than a single node with the maximal estimated degree in a neighborhood, and we need to somehow be able to break symmetry. If the number of neighbors with the maximal estimated degree is not too large, i.e., $O(\log{n})$, then we can break symmetry by sending IDs, using the simple BCC framework. However, if the number of neighbors with maximal estimated degree is larger, employing BCC will require too many rounds. Our crucial observation is that this event occurs with low probability, and hence we can simply disregard it. In more detail, a marked node that has the maximal estimated degree in its marked neighborhood estimates the number of its marked neighbors with maximal estimated degree. If this number is $O(\log{n})$ then symmetry is broken using IDs, sent using BCC. Otherwise, the node does not enter the MIS.

The pseudocode of our implementation is given in Algorithm~\ref{alg:luby}. The proof of correctness is straight forward using one additional verification round for each phase, where each node that is about to enter the MIS transmits this information to its neighbors, and if a conflict is detected then the conflicting nodes both give up their attempt to enter the MIS. Theorem~\ref{thm:luby} shows that w.h.p., the algorithm terminates after $O(\log{n})$ rounds.

\begin{algorithm}
Let $C$ be an $[N, O(\log^2{N}), O(\log{N})]$-BCC code\\
Initially: $V'=V$, $M=\emptyset$, $S=\emptyset$\\
Locally: For every node $v$, $\alpha(v,V') = (\delta(v,V'),ID_v)$\\
Repeat until $V'=\emptyset$:\\
\quad For every $v\in V'$ do:\\
\quad \quad $\delta(v,V')\leftarrow\AlgAppDeg(v,V')$ \label{line:appdeg}\\
\quad \quad If $\delta(v,V') = 0$ then $b(v) \leftarrow 1$\\
\quad \quad Else, $b(v) \leftarrow 1 \mbox{ w.p. } 1/ 2c \delta(v,V'); 0, \mbox{ otherwise}$.\\
\quad $S\leftarrow\{ u\in V' \mid b(u)=1\}$\\
\quad  For every $v\in S$ do:\\
\quad  \quad $\delta(v,S)\leftarrow\AlgAppDeg(v,S)$ \label{line:appdegmarked}\\
\quad  $S'\leftarrow\{u\in S \mid \delta(u,S)\leq \log{n}\}$ \label{line:Sprime}\\
\quad  For every $v\in S'$ do:\\
\quad  \quad Send $C[\langle v,\alpha(v,V')\rangle]$\\
\quad  \quad If $\alpha(v,V') > \max\{\alpha(u,V') \mid u\in (S'\cap\Gamma(v))\setminus\{v\}\}$, then $m(v)\leftarrow 1$ \label{line:maxdeg}\\
\quad  \quad Else, $m(v)\leftarrow 0$\\
\quad  \quad If $m(v)=1$ and $\{u \in \Gamma(v) \mid m(u)=1\} \neq\emptyset$ then $m(v)\leftarrow 0$ \label{line:verify}\\
\quad  $A \leftarrow \{u\in S' \mid m(u)=1\}$\\
\quad  $M\leftarrow M \cup A$\\
\quad  $S\leftarrow \emptyset$\\
\quad  $V'\leftarrow V' \setminus (M\cup \Gamma(M))$ \label{line:removed}\\
\quad  For every $u \in A$, Send $1$\\
\quad  For every $u \in \Gamma(A) \setminus A$, Send $0$\\

    \caption{An MIS algorithm for general graphs. The parameter $c$ is the constant given by \AlgAppDeg.}
    \label{alg:luby}
\end{algorithm}

\begin{thm}
\label{thm:luby}
The set $M$ returned by Algorithm~\ref{alg:luby} is an MIS. The algorithm completes in $O(\log{n})$ rounds, w.h.p.
\end{thm}

To prove that $M$ is an independent set, we claim that no two neighbors enter it in the same phase, which is guaranteed by Line~\ref{line:verify}. If a node was added to $M$ in a certain phase, then because only a single bit of information can be sent to it by all of its neighbors, it is removed along with all of its neighbors from any following phases, which gives that $M$ is indeed an independent set. It also holds that $M$ is maximal, since any node removed in a certain phase is either in $M$ or a neighbor of a node in $M$.

The main task is to prove the number of rounds required for the algorithm to complete. We follow the line of proof of Luby~\cite{Luby86}, and show that in each phase, a constant fraction of the edges touch at least one removed node. This implies that there are no more edges after $O(\log{n})$ rounds, w.h.p., after which the algorithm completes. The details of the proof are adapted from the proof of~\cite{WattenhoferCourse} to Luby's algorithm.

One modification we have to address is Line~\ref{line:Sprime}, where a node that has too many neighbors which are selected and have its maximal estimated degree simply drops out of the set of selected nodes. For a given node, the probability of being in $S\setminus S'$ is at  most $1/(2\log{n})^{\log{n}+1}$, since it needs to have a degree of at least $\log{n}$ and so do all of its at least $\log{n}$ selected neighbors. A union bound over all nodes still gives that this happens only in a very low probability, and hence the rest of the proof is conditioned on the evert that this does not occur.

\begin{lemma}
\label{lemma:joinM}
For every node $v$, the probability that $v$ is added to $M$ in a certain phase is at least $1/6c\delta(v,V')$.
\end{lemma}
\begin{proof}
To join $M$, a node $v$ has to first mark itself, and then be the maximal node that marked itself in its neighborhood. This implies that
\begin{equation}\label{eqn:pr-v}\Prob[v \in M]=\Prob[v \in M ~|~ b(v)=1]\Prob[b(v)=1]=\Prob[v \in M ~|~ b(v)=1]\cdot 1/2c\delta(v,V').\end{equation}
To bound this probability, we calculate the probability that $v$ does not enter $M$ despite being marked. This happens if either it does not have the maximal degree and ID among its marked neighbors, or if it has too many marked neighbors which share the largest degree with it.

We denote by $D(v)$ the neighbors of $v$ with the same approximate degree, that is $D(v)=\{u \in \Gamma(v) ~|~ \delta(u,V')=\delta(v,V')\}$. The probability that $v$ has too many marked neighbors which share its degree is small since in expectation it is constant, and w.h.p. it is at most $O(\log{n})$, using a standard Chernoff bound. Formally,
\begin{eqnarray*}
\EXP\big[\left|\{u \in D(v) ~|~ b(u)=1\}\right|\big] &=& \sum_{u \in D(v)}{1/2c\delta(u,V')} = \sum_{u \in D(v)}{1/2c\delta(v,V')}\\
&\leq& d(v,V')\cdot 1/2c\delta(v,V') \leq c\delta(v,V')/2c\delta(v,V') = 1/2.
\end{eqnarray*}
Since the random choices for $b(u)$ are independent for different nodes $u$, we have that
\[\Prob \big[|\{u \in D(v) ~|~ b(u)=1\}| > O(\log{n})\big] < 1/n^{t_1},\]
for a constant $t_1 > 1$.

Next, we bound the probability that a marked node $v$ does not have the maximal degree among its marked neighbors. Denote $D'(v) = \{u \in \Gamma(v) ~|~ \alpha(u,V') > \alpha(v,V')\}$. It holds that
\[\Prob[\exists u \in S'\cap D'(v)]  \leq \sum_{u \in D'(v)}\Prob[u \in S'] \leq \sum_{u \in D'(v)}1/2c\delta(v,V') \leq c\delta(v,V')/2c\delta(v,V')=1/2.\]

Finally, for every two neighbors $u,v$, the probability that $m(v)=m(u)=1$ at Line~\ref{line:verify} is at most $1/n^{t_2}$ for some constant $t_2$. This is because for this to occur, it must be that one of their degree estimations failed. This procedure is used at most twice per node and hence with probability at least $1-1/n^{t_2}$ all four invocations of this procedure were successful in obtaining a $c$-approximation, in which case either $m(v)$ or $m(u)$ are 0.

Hence, by Equation~\ref{eqn:pr-v}, the probability that $v$ joins $M$ is at least
\begin{eqnarray*}
\Prob[v \in M]&=&(1-\Prob[v \not\in M ~|~ b(v)=1])\cdot 1/2c\delta(v,V')\\
 &\geq& (1-(1/n^{t_1}+1/n^{t_2}+1/2))/2c\delta(v,V') \geq 1/6c\delta(v,V').
\end{eqnarray*}
\end{proof}

To show that a constant fraction of edges are removed in each phase, we show that a constant fraction of edges have at least one endpoint that is removed. We say that a node $v$ is \emph{good} if $\sum_{u \in \Gamma(v)}{1/2c\delta(u,V')} \geq 1/6c$, and claim that good nodes are removed with constant probability.

\begin{lemma}
\label{lemma:remove}
Let $v$ be a good node. Then the probability that $v$ gets removed in Line~\ref{line:removed} is at least $p_r = 1/18c - 1/18c^2$.
\end{lemma}
\begin{proof}
If there is a node $u \in \Gamma(v)$ such that $\delta(u,V') \leq 2$ then by Lemma~\ref{lemma:joinM} we have that with probability at least $1/6c\delta(u,V')\geq 1/12c \geq p_r$ the node $u$ joins $M$, and $v$ is removed in Line~\ref{line:removed}.

Otherwise, all nodes $u \in \Gamma(v)$ are such that $\delta(u,V') \geq 3$, and hence $1/2c\delta(u,V') \leq 1/6c$.
Let $L \subseteq \Gamma(v)$ be a subset of neighbors of $v$ such that $1/6c \leq \sum_{u \in L}{1/2c\delta(u,V')} \leq 1/3c$. The set $L$ exists because if we take all of $\Gamma(v)$ then since $v$ is good it holds that $\sum_{u \in \Gamma(v)}{1/2c\delta(u,V')} \geq 1/6c$, and if this sum is larger than $1/3c$ then we can take out nodes until we reach such a set $L$ (because for every $u \in \Gamma(v)$ we have $1/2c\delta(u,V') \leq 1/6c$).

We can now calculate the probability that $v$ is removed in Line~\ref{line:removed}, by being a neighbor of a node in $M$.

\begin{eqnarray*}
\Prob[v \in \Gamma(M)] &\geq& \Prob[\exists u \in L\cap M] \geq \sum_{u \in L}{\Prob[u \in M]} - \sum_{u,w \in L, u\neq w}{\Prob[u \in M \wedge w \in M]}\\
&\geq& \sum_{u \in L}{\Prob[u \in M]} - \sum_{u,w \in L}{\Prob[u \in S \wedge w \in S]}\\
&\geq& \sum_{u \in L}{1/6c\delta(u,V')} - \sum_{u,w \in L}{1/2c\delta(u,V')\cdot 1/2c\delta(w,V')}\\
&\geq& \sum_{u \in L}{1/2c\delta(u,V')\left(1/3 - \sum_{w \in L}{1/2c\delta(w,V')}\right)}\\
&\geq& 1/6c(1/3-1/3c) = 1/18c - 1/18c^2 = p_r.
\end{eqnarray*}
The third inequality above follows since the probability of a node to be in $M$ is at most its probability of being in $S$.
This completes the proof.
\end{proof}

Next, we consider a directed auxiliary graph $G'$ over $V'$, where each edge of the graph induced by $V'$ is directed towards the endpoint with the larger $\alpha(v,V')$ value. We claim that the outdegree in $G'$ of any node $v$ which is not good is at least twice its indegree. This holds because otherwise, $\sum_{u \in \Gamma(v)}{1/2c\delta(u,V')} \geq \sum_{u \in \Gamma(v)}{1/2c\delta(v,V')} \geq d(v)/3 \cdot 1/2c\delta(v,V') \geq 1/6$. It implies that at least half of the edges of the graph induced by $V'$ are good, in the sense that they have at least one good endpoint, since the number of edges that are directed from a non-good node to a good node is at least the number of edges in between two non-good nodes.

\begin{proofof}{Theorem~\ref{thm:luby}}
By Lemma~\ref{lemma:remove} each good node is removed in Line~\ref{line:removed} with probability at least $p_r$. Since at least half of the edges have a good endpoint, this implies that at least half of the edges have a probability of $p_r$ to be removed. Let $X_e$ be the characterizing random variable of the event that edge $e$ is removed. Then $\EXP[X_e] \geq p_r$ and by linearity of expectation we have that the expected number of edges that are removed is $p_r|E|$. Since phases are independent this implies termination in $O(\log{n})$ phases, in expectation. To show that this also holds with high probability, we let $Y_i$ be the characterizing random variable of the event that at least $p_r|E|$ edges were removed in phase $i$, and we denote $Y=\sum{Y_i}$. The above argument implies that $\EXP[Y] = O(\log{n})$. Since phases are independent we can use a standard Chernoff bound to get that $\Prob[Y > O(\EXP[Y])] < 1/n^{t_3}$ for some constant $t_3$. Hence, with high probability, the number of phases required is $O(\log{n})$.
\end{proofof}

\vspace{-0.5cm}

\subsection{Coloring}
\label{app:coloring}

\def\APPENDCOL{
In this Section we address the problem of finding a $(\Delta+1)$-coloring of the underlying graph. Each node has to output a color in $\{1,...,\Delta+1\}$, such that no two neighbors share the same color. We build upon known techniques and embed the usage of BCC codes to them in order to obtain our results.

Suppose $A$ is an MIS algorithm that works in $f(n)$ rounds in BGG graphs. 
We show how to get an algorithm for $(\Delta+1)$-coloring in $O(\Delta+f(n))$ rounds for BGG graphs. The algorithm follows the line of the $O(\log^*n)$ algorithm for $(\Delta+1)$-coloring in BGG graphs in the message-passing model, by Schneider and Wattenhofer~\cite{SchneiderW2010}. We repeat the following procedure for the subgraph $G_i$ induced by the node set $V_0$, where initially $V_0=V$.\\

\noindent 1. Find an MIS $S_i$ in the graph $G_i$.\\
2. Denote by $H_i$ the graph $(S_i,E_{S_i})$, where $(u,v) \in E_{S_i}$ if $d_{G_i}(u,v) \leq 3$. Find an MIS $S'_i $ in $H_i$.\\
3. Each node in $S'_i$ colors itself and all of its neighbors.\\
4. $V_{i+1} = V_i - \Gamma(S'_i)$

\paragraph{Analysis Sketch:}
The analysis follows the analysis of previous work and hence we do not repeat it here in full, but rather sketch the idea.
There are a constant number of phases because in each neighborhood of radius 6 there is at least one node in $S'_i$ which gets colored in phase $i$, and there can be at most a constant number of such nodes throughout the phases since they are independent.

Steps 1 and 2 take at most $f(n)$ rounds.
Step 3 requires $O(\Delta)$ rounds, as follows. Each node $v \in S'_i$ transmits its ID. Each uncolored neighbor $u$ of $v$ applies for a color, by sending its ID with some probability. Once $v$ knows an ID of a neighbor $u$ it sends that ID and then $u$ colors itself by announcing its color (this is similarly to~\cite{SchneiderW09}).
This way all nodes know about the colors that are already used by their neighbors, and allows them to safely choose an unused color.

It may be possible to go below $O(\Delta)$ and obtain a solution in $O(\Delta/\log{n})$ rounds, since one can use BCC for the randomized attempts of coloring the neighbors.
Moreover, we conjecture that one can derive a lower bound of $\Omega(\Delta/\log{n})$ using an information theoretic argument that is similar to that of Schneider and Wattenhofer~\cite{SchneiderW2010b}.
}

\APPENDCOL

\vspace{-0.2cm}
\subsection{Minimum Dominating Set}
\label{app:dom}
\APPENDDOM

\commlongend

\end{document}